\documentclass[12pt,english]{article}
\usepackage[T1]{fontenc}
\usepackage[latin9]{inputenc}
\usepackage{geometry}
\usepackage[nolists,figuresonly]{endfloat}
\usepackage{babel}
\usepackage{verbatim}
\usepackage{url}
\usepackage{amsmath}
\usepackage{amsthm}
\usepackage{amssymb}
\usepackage{graphicx}
\usepackage{rotfloat}
\usepackage{setspace}
\usepackage[authoryear]{natbib}
\usepackage[colorlinks, citecolor=blue]{hyperref}
\usepackage{xcolor}

\makeatletter
\theoremstyle{plain}
\newtheorem{lyxalgorithm}{\protect\algorithmname}
\theoremstyle{plain}
\newtheorem{lem}{\protect\lemmaname}
\theoremstyle{remark}
\newtheorem{rem}{\protect\remarkname}
\theoremstyle{definition}
 \newtheorem{example}{\protect\examplename}
\theoremstyle{plain}
\newtheorem{prop}{\protect\propositionname}
\theoremstyle{plain}
\newtheorem{thm}{\protect\theoremname}
\theoremstyle{plain}
\newtheorem{cor}{\protect\corollaryname}
\newenvironment{lyxlist}[1]
	{\begin{list}{}
		{\settowidth{\labelwidth}{#1}
		 \setlength{\leftmargin}{\labelwidth}
		 \addtolength{\leftmargin}{\labelsep}
		 }}
	{\end{list}}

\usepackage{appendix}
\usepackage{multirow}

\usepackage{chngcntr}
\usepackage{apptools}
\AtAppendix{\counterwithin{corollary}{section}}

\usepackage{subcaption}
\makeatother

\providecommand{\algorithmname}{Algorithm}
\providecommand{\corollaryname}{Corollary}
\providecommand{\examplename}{Example}
\providecommand{\lemmaname}{Lemma}
\providecommand{\propositionname}{Proposition}
\providecommand{\remarkname}{Remark}
\providecommand{\theoremname}{Theorem}

\renewcommand{\hat}{\widehat}

\allowdisplaybreaks

\defcitealias{phillips2021business}{PJ}	
\defcitealias{phillips2021boosting}{PS}

\DeclareDelayedFloatFlavor{sidewaysfigure}{figure}

\title{The boosted HP filter is more general \\than you might think\thanks{Phillips acknowledges research support from a Kelly Fellowship at the
	University of Auckland. 
Shi thanks the National Natural Science Foundation of China (NSFC) for financial support under the grant number 72133002.
	\small{Ziwei Mei: \texttt{zwmei@link.cuhk.edu.hk}; Peter C.B.~Phillips: \texttt{peter.phillips@yale.edu}; Zhentao Shi (corresponding author): \texttt{zhentao.shi@cuhk.edu.hk}. 934 Esther Lee Building, the Chinese University of Hong Kong, Shatin, New Territories, Hong Kong SAR, China. }}}

\author{Ziwei Mei$^{a}$, Peter C. B. Phillips$^{b,c,d}$ and Zhentao Shi$^{a}$ \vspace{2mm}
	\\ \small $^{a}$The Chinese University of Hong Kong 
	\\ \small $^{b}$University of Auckland, $^{c}$Yale University, 
	\\ \small $^{d}$Singapore Management University 
	}

\date{}

\usepackage{xr}
\makeatletter

\newcommand*{\addFileDependency}[1]{
\typeout{(#1)}
\@addtofilelist{#1}
\IfFileExists{#1}{}{\typeout{No file #1.}}
}\makeatother

\newcommand*{\myexternaldocument}[1]{%
\externaldocument{#1}%
\addFileDependency{#1.tex}%
\addFileDependency{#1.aux}%
}
\myexternaldocument{JAE_2024_Supplement}

\begin{document}


\maketitle

\thispagestyle{empty} 

\begin{center}
\textbf{\large{}Abstract}{\large\par}
\par\end{center}

\small{
The global financial crisis and Covid recession have renewed discussion concerning trend-cycle discovery in macroeconomic data, and boosting has recently upgraded the popular HP filter to a modern machine learning device suited to data-rich and rapid computational environments. This paper extends boosting's trend determination capability to higher order integrated processes and time series with roots that are local to unity. The theory is established by understanding the asymptotic effect of boosting on a simple exponential function. Given a universe of time series in FRED databases that exhibit various dynamic patterns, boosting timely captures downturns at crises and recoveries that follow.
}\bigskip


\vspace{1mm}

\newpage{}

\normalsize
\onehalfspacing

\doublespacing
\section{Introduction\label{sec:Introduction}}

Understanding long-term trends and short-term business cycles in economic activity 
was a primary pursuit of the foundational researchers of the econometrics profession especially during the years of the Great Depression where these matters figured prominently and prompted the development of new econometric approaches \citep{ragnar1933propagation,tinbergen1939business}.
At that time with very few exceptions data were scarce. By contrast vast datasets are now available to researchers covering multiple decades of quarterly and monthly time series observations across a wide range of economic variables. Long trajectories of data provide rich information about many aspects of economic activity and wellbeing, including the impact of technical progress on growth and the course and consequences of intermittent slowdowns and recessions. Analysis of the information carried in such time series provides useful indicators of the passage and present state of economic activity, which in turn helps to shape assessments of policymakers, regulators, corporate executives, and consumers in guiding decision making. The 2008 global financial
crisis (GFC) and its aftermath and most recently the global economic impact of the Covid-19 pandemic are timely reminders that the work launched by \cite{ragnar1933propagation} and \cite{tinbergen1939business} remains an ongoing mission for the econometrics community. 

Modern econometric approaches often focus on decomposing time series observations of a variable $x_{t}$
into additive components that represent long run trending behavior, $f_t$, and cyclical activity, $c_t$,  as
\begin{equation}
x_{t}=f_{t}+c_{t}.\label{eq:y_decomp}
\end{equation}
The trend embodies the long run general course or tendency in the data and the cycle reflects periodic fluctuations in economic activity in which businesses, labor markets and consumer behavior alternately expand and contract. Trends in many economic aggregates like a nation's real GDP are primarily determined by the impact of production technologies, the size and quality of labor forces, the accumulation of physical and human capital, and entrepreneurship, whereas business cycles are affected by shorter term internal and external influences including contractionary or expansionary fiscal, monetary, and political policies, combined with the dynamic propagation of these forces within an economy. As such, trend and cycle may be considered latent elements in the data which are to be identified and estimated by econometric methods through decomposition or direct modeling. 

Since its introduction the Hodrick-Prescott (HP) filter \citep{hodrick1997postwar} has become a convenient and highly popular off-the-shelf choice for trend-cycle decomposition. 
Given an observed time series $x=(x_{1},x_{2},\dots,x_{n})^{\prime}$,
the HP filter finds $\widehat{f}^{{\rm HP}}=(\widehat{f}_{1}^{{\rm HP}},\widehat{f}_{2}^{{\rm HP}},\ldots,\widehat{f}_{n}^{{\rm HP}})^{\prime}$
via a penalized least squares criterion 
\begin{equation}
\text{\ensuremath{\widehat{f}}}^{{\rm HP}}=\arg\min\limits _{(f_{t})}\left\{ \sum_{t=1}^{n}(x_{t}-f_{t})^{2}+\lambda\sum_{t=3}^{n}(\Delta^{2}f_{t})^{2}\right\} ,\label{eq:hpfilter}
\end{equation}
where the second difference $\Delta^{2}f_{t}=\Delta f_{t}-\Delta f_{t-1}=f_{t}-2f_{t-1}+f_{t-2}$ of the trend component provides a measure of fluctuations, whose degree is controlled through the tuning parameter $\lambda \in (0,\infty)$ that governs the extent of the penalty in the second component of the extremum criterion \eqref{eq:hpfilter}. Penalization plays a key role in determining the outcome of the filter, with larger $\lambda$ imposing a greater penalty on roughness thereby favoring smoothness in the outcome $\widehat{f}^{{\rm HP}}$. 

Formally developed a century ago by \cite{whittaker1922new} in pathbreaking work on penalized estimation, 
ideas for `graduating' data have a long history in actuarial science and subsequent work in statistics and engineering, 
as reviewed in \citet[hereafter, PJ]{phillips2021business}.
When $\lambda \to \infty$
the solution for \eqref{eq:hpfilter} satisfies $\Delta^{2}f_{t}=0$, giving a linear trend function $f_t=a+bt$ for some constant coefficients $a$ and $b$. The conventional choice for quarterly data is $\lambda=1600$, suggested by \citet{hodrick1997postwar} based on empirical experimentation with macroeconomic time series. Corresponding settings of $\lambda$ for empirical work with monthly and annual data were given in \citet{ravn2002adjusting}.

A poignant newspaper column posted by \citet{krugman2012} concerning the importance of long run trend identification during the GFC raised substantial interest on trend determination and the HP filter, spurring an influx of opinions, theory investigations, novel proposals, and empirical evidence. \citet{de2016econometrics}, \citet{cornea2017explicit},
and \citet{sakarya2017property} explored the HP filter's finite-sample
algebraic properties. Arguing that the HP filter is a nonparametric procedure in which tuning parameters are typically sample size dependent to achieve consistent estimation, \citetalias{phillips2021business}
analyzed its asymptotic features via operator calculus, showing that the common
setting $\lambda=1600$ is too large to completely remove stochastic trends in time series of the length usually encountered in empirical work. Instead, the HP filter smooths those trends into paths that are asymptotically differentiable forms of Brownian motion. Several revised or modified versions based on the HP filter have recently emerged. For example, \citet{yamada2020smoothing} and \citet{yamada2021trend}
generalize the HP filter to overcome data imperfections, and with
the advent of Covid-19 \citet{lee2021sparse} use a sparse filter to identify the turning points in the inflection rates of the virus.

Retaining the squared penalization scheme of the original
HP filter, \citet[hereafter, PS]{phillips2021boosting} proposed a \emph{boosted}
HP filter (bHP hereafter) designed to upgrade the procedure to a machine
learning device that uses the data more intensively to improve its properties and performance. The bHP filter is a repeated application of the
HP filter to the residual extracted in the last iteration (see Algorithm \ref{alg:Boosted-HP-filter} below) in which the number of iterations, $m$, controls the intensity of reusage. In practice, \citetalias{phillips2021boosting} suggested monitoring a stopping criteria
to terminate the iteration in a data-driven manner (see Algorithm
\ref{alg:Automated-bHP-BIC}), which makes bHP automated in application, as
envisaged in \citet{phillips2005automated}. In three empirical
examples, \citetalias{phillips2021boosting} fed 127 individual time series of various lengths and
trending patterns through the bHP machine. The trend and cycle estimates returned after a few iterations in this algorithm largely confirmed and refined the existing findings in the literature for these series. Further simulation experiments and empirical applications (e.g., \cite{hall2021does,hall2024selecting}) have been conducted and the robustness continues to hold. 

These outcomes are indicative that, as a machine learning device which is agnostic  about the data generating mechanism, the bHP filter satisfactorily accommodates trend processes in applications that are more general than unit root processes.
But there remain gaps in the present asymptotic theory that involve persistent time series beyond unit root stochastic trends, such as near unit root processes, for which it would be useful in empirical work to have a theory foundation for the use of boosting in consistent trend determination. This paper offers an affirmative answer with analytic, simulation, and empirical support for this extension.

In doing so we provide a preparatory result (Lemma \ref{lem:m-operator}
in Section \ref{subsec:Operator-Form}) that characterizes the shrinkage effect of
the operational form of the HP filter. 
We keep using
the tuning parameter formula $\lambda=\mu n^{4}$ for some constant $\mu>0$,
which is an expansion rate extensively studied in \citetalias{phillips2021business}. For sample sizes
that are typical in quarterly economic data, this expansion rate approximates
well the actual form of the filter in practical work with the common
choice $\lambda=1600$ for the HP filter. In view of the two-sided nature of the filter, the HP operator is a function
of lead and lag operators. For a class of complex numbers $a\in\mathbb{C}$
such that $a^{4}$ is a non-negative real number, Lemma \ref{lem:m-operator}
shows that the HP residual operator shrinks the complex
exponential $\mathrm{e}^{ax}$ towards zero by the factor $\mu a^{4}/ (\mu a^{4}+1 )\in[0,1)$,
which is a pseudo-differential operator extension of the elementary property $D_x^{m}\mathrm{e}^{ax}=a^{m}\mathrm{e}^{ax}$ of the usual differential operator $D_x=\mathrm{d}/\mathrm{d}x$.  Repeating the operation $m$ times gives the power factor $(\mu a^{4}/ (\mu a^{4}+1 ) )^{m}\mathrm{e}^{ax}$, which tends to zero as $m \to \infty$. 

This lemma enables a unified development of asymptotics of the HP and bHP filters for a variety of nonstationary trend processes that include unit root $I(1)$ time series, higher order integrated $I(q)$ processes with integer $q\geq2$, and local-to-unity (LUR) processes, which
are among the most widely used models for nonstationary data. In
each case, boosting enables consistent estimation
of the trend whereas single implementation of the HP filter is inconsistent, producing
a smoothed version of the original trend process. Upon standardization these trends all have asymptotic stochastic process representations in terms of convergent series of trigonometric functions with complex exponential forms that are amenable to analysis by operator methods and thereby deliver the respective asymptotic forms of the filter operation. Similar methods apply  to general deterministic trend functions with convergent Fourier series representations. Taken together, the asymptotic results span a wide range of trend models commonly used in econometric practice.      

The present paper continues the line of research by \citetalias{phillips2021business} and \citetalias{phillips2021boosting} 
in understanding of the capabilities of the (boosted) HP filter.
While these two papers worked with unit root stochastic trends exclusively, 
this paper considerably expands the types of stochastic and deterministic trend mechanisms.
The unified analytical framework sheds new theoretical insight about the shrinkage effect of the HP operator,
which exposes the source of limitations of the HP filter and strengths of the repeated fitting.
It complements ongoing work that considers the use of boosting methods for time series with long range dependence \citep{biswas2022longrange}. In addition, the present work provides further numerical evidence of the robustness and versatility of bHP in simulations and various real data applications. Overall, the findings reveal analytic properties and empirical performance that support the boosted filter as a useful machine learning method for extracting low frequency components of macroeconomic time series, thereby contributing to the comprehension of trend phenomena.

In machine learning terminology the HP and bHP filters are \emph{unsupervised
learning} methods which seek to extract key features of the data but do not
use regressors to fit the dependent variables or `labels'. In
contrast to this methodology, \citet{hamilton2018you} firmly advocated that the HP filter
should be replaced in empirical work by an autoregression $x_{t+h} = \beta_0 + \beta_1 x_{t} + \beta_2 x_{t-1} + \dots + \beta_p x_{t-p+1} + v_{t+h}$
with specific order $(h,p)=(8,4)$ recommended for quarterly data. The trend component at the period $t+h$ is estimated by the fitted dependent variable $\hat x_{t+h}$. We call it Hamilton's regression filter (HRF). Regressions of this type fall into the category of simple \emph{supervised learning} methods in which
a few lagged observations are trained to predict a future target. Unlike
nonparametric approaches such as HP and bHP where tuning parameters
are unavoidable and play a central role in consistent estimation, parametric autoregressions typically bypass tuning parameters, although users still need to decide on the number of lags \citep{quast2022reliable}.
\cite{hamilton2018you}'s paper has stirred considerable discussion and debate. The issues raised bear directly on empirical econometric practice and they affect economic policy analysis in fundamental ways concerning the manner in which observed economic indicators can be interpreted as indicative of long term trend behavior as distinct from cyclical fluctuations, the very considerations that motivated \citet{krugman2012}'s public post. \citet{schuler2021cyclical} pointed out that HRF fails to reproduce the standard chronology of US business cycles and `emphasizes cycles that exceed the duration of regular business cycles (i.e., longer than 8 years), and completely mutes certain short-term fluctuations.'

Among the commentary, \citet{knightboosted} provided a frequency domain analysis of both the HP and bHP filters, showing that the latter has some additional `free pass' effects at low frequencies over that of the HP filter, giving it improved recovery properties for trends with frequencies in an interval around the origin. 
Recent empirical works by \citet{drehmann2018you}, \citet{hall2021does} and \citet{jonsson2020cyclical} compared the
HP filter and HRF in real data examples and the accumulated empirical
evidence from these studies favors the former.
In subsequent work \citet{hall2024selecting} provided transfer function analysis and studied the empirical performance of the HP and bHP filters with various stopping rules, 
recommending a `twicing' version of bHP with a second iteration (2HP) for trend and growth cycle analysis with New Zealand quarterly macroeconomic data.  
\citet{lu2023boost} extended the discussion of the properties of the HP filter and its boosted version in line with their finite sample weights and linked the necessity of boosting to the patterns observed in the 
macro time series. 

\citetalias{phillips2021boosting} provided a detailed response to \citet{hamilton2018you}'s critiques of the HP filter and advocated the use of an automated bHP filter.
The present paper adds further simulation and real data testimony concerning these two different approaches. 
We apply the HP-based methods and HRF to the open-source FRED-QD
database \citep{mccracken2020fred} and FRED-MD database \citep{mccracken2016fred}.
Our empirical findings provide a fairly consistent message about the performance of the HP and bHP filters in relation to the regression filter.
Both a small-scale
test drive of the procedures on US real GDP and a large-scale deployment to the entire databases are conducted. In brief, the HP filter is able
to capture most historical 
business cycles although the fitted trends tend to
be overly flattened or smoothed. The bHP filter is more adaptive to various generating mechanisms and
patterns in the trend. On the other hand, HRFs typically seek to reduce observed series to martingale difference residuals and to offer useful mechanisms for prediction and impulse response analysis, but do not provide a method for identifying and estimating general trend and cycle processes, particularly those for which irregularity is a prominent feature,
thereby failing to capture key trend and cycle elements of a macroeconomy.

bHP is an $L_{2}$-boosting method applied to time series, drawing on ideas of boosting by repetitive use in the computer science literature \citep{freund1995desicion} with statistical roots that go back to \citet{tukey1977exploratory}'s introduction of the \emph{twicing} technique. The key notion is to gradually `boost' an
ensemble of many weak learners
into a more powerful fitting machine.
Boosting has evolved into a very 
successful machine learning method, with many variants proposed along
the way, for instance \emph{adaboost} \citep{friedman2000additive},
componentwise boosting \citep{buehlmann2006boosting}, and $L_{2}$-boosting
\citep{buhlmann2003boosting}.

In addition to the above, some useful theoretical developments and applications of boosting have occurred in the econometric literature \citep{bai2009boosting,shi2016econometric,yousuf2021boosting,kueck2023estimation}.
\citetalias{phillips2021boosting} was the first paper to employ boosting methods in nonstationary time series, much of its asymptotic analysis being built on the foundation of functional limit theory \citep{phillips1986understanding,phillips1987time,phillips1987towards}. Careful interpretation of regression findings with nonstationary time series and panels has relied on orthonormal series representations of limiting stochastic processes \citep{phillips1998new}. These methods offer an understanding of nonstationarity in terms of coordinate basis functions and have, in turn, proved useful in analyzing the asymptotic properties of both the HP and bHP filters. Recent years have also witnessed the proliferation in econometric work of machine learning methods, including cross-sectional studies \citep{belloni2014inference,caner2018asymptotically,farrell2021deep,athey2021matrix},
time series \citep{shi2022L,masini2021counterfactual,babii2022machine},
and panel modeling \citep{su2016identifying,moon2018nuclear}, to name a few.

The rest of the paper is organized as follows. 
Section \ref{sec:Filter} introduces the HP and bHP filters and their operational
forms,  and then presents the basic asymptotic approximation.
It is followed in Section \ref{sec:theo-Applications} by three
applications of the theory to unit root, higher order integrated, and LUR time series on
the smoothing properties of the usual HP filter and the consistency of bHP.
The theory is supported by simulation exercises
in Section \ref{sec:Simulations}. 
Section \ref{sec:emp} demonstration with HP, bHP and
HRF methods applied to US quarterly GDP, and then  
provides a `big data' empirical implementation of the methods that explores FRED-QD or -MD database to study trends and cycles in the economy. Proofs and additional simulations are given in the Appendix as an online supplementary document.

\section{Smoothing effects of the HP filter \label{sec:Filter}}

The squared penalty in the HP filter (\ref{eq:hpfilter}) leads to the explicit solution
\[
\hat{f}^{\text{HP}}=Sx,
\]
where $S=(I_{n}+\lambda D_{n}D_{n}^{\prime})^{-1}$, $D_{n}^{\prime}$
is the rectangular $(n-2)\times n$ matrix with the second differencing
vector $d=(1,-2,1)^{\prime}$ along the leading tri-diagonals and
$I_{n}$ is the $n\times n$ identity matrix. The fitted cyclical
component (residual) is $\hat{c}^{\text{HP}}=(\hat{c}_{1}^{\text{HP}},\hat{c}_{2}^{\text{HP}},\ldots,\hat{c}_{n}^{\text{HP}})'=x-\hat{f}^{\text{HP}}$.
These compact expressions are useful in studying the effects of boosting.
\begin{lyxalgorithm}
\label{alg:Boosted-HP-filter} The boosted HP filter (in $m$ iterations)
\begin{description}
\item [{Step1}] Specify the smoothing parameter $\lambda>0$ and the number
of iterations $m\geq2$. Set $\widehat{c}^{(1)}=\widehat{c}^{\mathrm{HP}}$
and $j=1$. 
\item [{Step2}] If $j<m$, set $\widehat{c}^{(j+1)}=(I_{n}-S)\widehat{c}^{\left(j\right)}$ and then update $j = j+1$. 
\item [{Step3}] Repeat Step2 until $j=m$. Save the fitted cyclical component
$\widehat{c}^{\left(m\right)}$ and the estimated trend $\widehat{f}^{\left(m\right)}=y-\widehat{c}^{\left(m\right)}.$
\end{description}
\end{lyxalgorithm}
\noindent The recursive form of the above algorithm is simply  
$\widehat{c}^{\left(m\right)} =(I_{n}-S)^{m}x$ and $\widehat{f}^{\left(m\right)}  =\left[I_{n}-(I_{n}-S)^{m}\right]x$. 
The bHP filter is easy to implement. For instance, if the original HP filter is called by the
\texttt{hpfilter} function in the \texttt{R} package \texttt{mFilter},
then bHP with $m$ iterations can be carried out with a single line
of code in the following \texttt{tidyverse} style
\begin{quotation}
\texttt{purrr::rerun(m, y \%<>\% mFilter::hpfilter(lambda, \textquotedbl lambda\textquotedbl )
\%>\% .\$cycle)}
\end{quotation}
where \texttt{y} is the observed time series, \texttt{m} controls the number
of iterations, and \texttt{lambda} sets the tuning parameter.

\subsection{Shrinkage Effects\label{subsec:Operator-Form}}

With the tuning parameter $\lambda=\mu n^{4}$ for some constant
$\mu>0$ and lag operator $\mathbb{L}$, the asymptotic approximation (as $n\to\infty$)
of the HP filtered trend has the operator form
\[
G_{\lambda}=\dfrac{1}{\lambda\mathbb{L}^{-2}(1-\mathbb{L})^{4}+1}=\dfrac{1}{\mu\mathbb{L}^{-2}\left[n(1-\mathbb{L})\right]^{4}+1},
\]
and the corresponding residual operator is
$
1-G_{\lambda}=\dfrac{\mu\mathbb{L}^{-2}\left[n(1-\mathbb{L})\right]^{4}}{\mu\mathbb{L}^{-2}\left[n(1-\mathbb{L})\right]^{4}+1}.
$
The estimated trends of the HP and bHP ($m$ iterations) filters then have the asymptotic forms $\widehat{f}_{t}^{\text{HP}}=G_{\lambda}x_{t}$ and $\widehat{f}_{t}^{(m)}=\left[1-\left(1-G_{\lambda}\right)^{m}\right]x_{t}$, respectively.

The following lemma reveals the effect of the operator when it
is applied to a class of complex exponential functions of the form $\exp\left(at/n\right)$.
\begin{lem}
\label{lem:m-operator} 
Let $\lambda=\mu n^{4}$ for some real constant $\mu>0$, and $t\in\mathbb{N}:=\left\{ 1,2,\ldots\right\}$. For a real number $z>1$, define a set $\mathcal{A}(z):=\left\{ a\in\mathbb{C}:a^{4}\in[0,(\log z)^{2}]\right\}$.
\begin{enumerate}
\item For any fixed $m\in\mathbb{N}$, as $n\to\infty$
we have
\begin{equation}
\sup_{1\leq t\leq n,\,a\in\mathcal{A}(n)}\left|\left[\left(1-G_{\lambda}\right)^{m}-\left(\frac{\mu a^{4}}{\mu a^{4}+1}\right)^{m}\right]\mathrm{e}^{at/n}\right|\to 0. \label{eq:mconv_e}
\end{equation}
\item As $m,n\to\infty$ we have 
\begin{equation*}
\sup_{1\leq t\leq n,\,a\in\mathcal{A}(n\wedge m)}\left|\left(1-G_{\lambda}\right)^{m}\mathrm{e}^{at/n}\right|\to 0,
\end{equation*}
where $n\wedge m:=\min\left\{ n,m\right\} $. 
\end{enumerate}
\end{lem}

Lemma \ref{lem:m-operator} (a) shows that
the operator $\left(1-G_{\lambda}\right)$ works as if $\mathrm{e}^{at/n}$
is multiplied by a real factor $\frac{\mu a^{4}}{\mu a^{4}+1}\in[0,1)$
at each operation,
because $a^{4} >0$ for any $a\in \mathcal{A}(n)$. 
Part (b) extends this approximation to include large $m$, where the factor
$\left(\frac{\mu a^{4}}{\mu a^{4}+1}\right)^{m}$ vanishes as $m \to \infty$.

\medskip

\begin{rem}\label{rem:m-operator} 
The set $\mathcal{A}(z)$ is a theoretical
artifact needed in the proofs; it is irrelevant to the data applications. 
The logarithmic rate of expansion $\log z$ in its definition is devised to control the complexity of this set relative to the sample size to make uniform convergence achievable.
For a finite $m$ in Part (a), the upper bound is
$|a^{2}| \leq \log n$.
When $m$ passes to infinity in Part (b),
the effect of the iterations of the operator is restricted into $|a^{2}| \leq \log\left(n\wedge m\right)$.
\end{rem}

Lemma \ref{lem:m-operator} shows the asymptotic effect of the operator $\left(1-G_{\lambda}\right)^m$ on a simple exponential function. \citetalias{phillips2021business} used the exponential function as an intermediate
step in studying the effect of the HP operator involving $\left(1-G_{\lambda}\right)$.
In the numerator of the residual operator $(1-G_{\lambda})$,
asymptotically when $n\to\infty$ 
the scaled differencing operator $\left[n\left(1-\mathbb{L}\right)\right]$
acts like a differential operator
and the lag operator $\mathbb{L}^{-1}$ acts like the
identity, so the effect of these operations applied to $\mathrm{e}^{at/n}$ is straightforward. The HP operator involves these elementary operators in the nonlinear operator
 $[\mu\mathbb{L}^{-2}\left[n(1-\mathbb{L})\right]^{4}+1]^{-1}$
and \citetalias{phillips2021business} showed how this complex operator may be analyzed as a pseudo-differential operator. Using this machinery \citetalias{phillips2021boosting} studied repeated
applications of the operator to the trigonometric functions in the Karhunen-Lo\`eve (KL) representation of Brownian motion --- see (\ref{eq:KLI1}) below. This approach is used in developing the asymptotic analysis in the next section.

\section{Applications}\label{sec:theo-Applications}

This section illustrates the use of Lemma \ref{lem:m-operator}
for trends that involve unit roots, higher order integrated time series and local unit root time series.

\subsection{Unit Root Processes}\label{subsec:Case-0:-Unit}

To demonstrate how the bHP filter moderates the residual component in the trend fitting process, we begin with simple unit root time series that was fully analyzed in \citetalias{phillips2021boosting}. The following discussion is heuristic to reveal the manner in which the moderation operates. Consider an observed time
series $x_{t}$ generated as an $I(1)$ stochastic trend from a unit root process $f_{t}=f_{t-1}+u_{t}$ with initial value $f_0=o_p(\sqrt{n})$ plus a stationary cyclical component $c_t$. 
The scale-normalized series satisfies the functional law \citep{phillips1992asymptotics} 
\begin{equation}
X_{n}(\cdot):=n^{-1/2}x_{t=\lfloor n\cdot\rfloor}\rightsquigarrow B(\cdot),\label{convergenceI1}
\end{equation}
where ``$\rightsquigarrow$'' denotes weak convergence in the relevant probability space, $B(\cdot)=BM(\omega^{2})$ is a Brownian motion with (long run)
variance $\omega^{2}$, and $\lfloor\cdot\rfloor$ is the integer
floor function. The KL representation of this Brownian motion over
the interval {[}0,1{]} is
\begin{equation}
B(r)=\sqrt{2}\sum\limits _{k=1}^{\infty}\dfrac{\sin\left[(k-\frac{1}{2})\pi r\right]}{(k-\frac{1}{2})\pi}\xi_{k}=\sum\limits _{k=1}^{\infty}\sqrt{\lambda_{k}}\varphi_{k}(r)\xi_{k},\label{eq:KLI1}
\end{equation}
where $\xi_{k}\sim i.i.d.\ N(0,\omega^{2})$ are the random coefficients,
$\lambda_{k}=1/[(k-\frac{1}{2})\pi]^{2}$ are the eigenvalues, and
$\{\varphi_{k}(r)=\sqrt{2}\sin[(k-\frac{1}{2})\pi r]=\sqrt{2}\sin(r/\sqrt{\lambda_{k}})\}_{k=1}^{\infty}$
is an orthonormal system of corresponding eigenfunctions. The series \eqref{eq:KLI1} converges almost surely and uniformly over $r \in [0,1]$.

When the innovation $u_{t}$ follows a linear process as in
\begin{equation}
u_{t}=C(\mathbb{L})\varepsilon_{t}=\sum_{j=0}^{\infty}c_{j}\varepsilon_{t-j},\ \ \sum_{j=0}^{\infty}j\left|c_{j}\right|<\infty,\quad C(1)\neq0\label{eq:wold}
\end{equation}
with $C(z)=\sum_{j=0}^{\infty}c_{j} z^j$, $\varepsilon_{t}=i.i.d.\left(0,\sigma_{\varepsilon}^{2}\right)$
and $E\left(\left|\varepsilon_{t}\right|^{p}\right)<\infty$ for some
$p>4$, we construct an expanded probability space with a Brownian motion $B(\cdot)$ for which uniform convergence holds almost surely \citep[Lemma 3.1]{phillips2007unit}, viz., 
\begin{equation}
\sup_{0\leq t\leq n}\left|\frac{x_{t}}{\sqrt{n}}-B\left(\frac{t}{n}\right)\right|=o_{a.s.}\left(\frac{1}{n^{1/2-1/p}}\right).\label{uniformI}
\end{equation}
In this space the convergence (\ref{convergenceI1}) takes the strong form 
\begin{equation*}
n^{-1/2}x_{\lfloor nr\rfloor}-B(r)=o_{a.s.}(1).
\end{equation*}
In what follows and unless otherwise stated, we assume that we are
working in this expanded probability space. In the original space
the results translate, as usual, into weak convergence mirroring (\ref{convergenceI1}). 

Write $\varphi_{k}(t/n)=\sqrt{2}\sin(\frac{t/n}{\sqrt{\lambda_{k}}})=\sqrt{2}\:\textrm{Im}
\left[\exp(\frac{\mathbf{i}t/n}{\sqrt{\lambda_{k}}})\right]$, 
where $\mathbf{i}:=\sqrt{-1}$ is the imaginary unit
and $\textrm{Im}[\cdot]$ gives the imaginary part of the argument. 
When the operator $\left(1-G_{\lambda}\right)^{m}$ is applied to
the $k$th term of (\ref{eq:KLI1}) for any fixed $k$, Lemma \ref{lem:m-operator}
gives 
\begin{equation}
\left(1-G_{\lambda}\right)^{m}
\textrm{Im}\left[{\rm e}^{\frac{\mathbf{i}t/n}{\sqrt{\lambda_{k}}}}\right]
\approx\left[\dfrac{\mu}{\mu+\lambda_{k}^{2}}\right]^{m}
\textrm{Im}\left[{\rm e}^{\frac{\mathbf{i}t/n}{\sqrt{\lambda_{k}}}}\right]
=\left[\dfrac{\mu}{\mu+\lambda_{k}^{2}}\right]^{m}\sin\left(\frac{t/n}{\sqrt{\lambda_{k}}}\right), \label{eq:heur1}
\end{equation}
so that when $m$ is large we have
\begin{equation}
\left(1-G_{\lambda}\right)^{m}\sin\left(\frac{t/n}{\sqrt{\lambda_{k}}}\right)\approx\exp\left(-\frac{m\lambda_{k}^{2}}{\mu+\lambda_{k}^{2}}\right)
\sin\left(\frac{t/n}{\sqrt{\lambda_{k}}}\right) \to 0, \label{eq:heurstic}
\end{equation}
as $m$ and $n$ pass to infinity. With careful handling of a finite-term
approximation to the infinite series in (\ref{eq:KLI1}), 
\citetalias{phillips2021boosting} showed that when $\lambda=\mu n^{4}$ the residual
of the bHP filter becomes 
\[
n^{-1/2}\widehat{c}_{\left\lfloor nr\right\rfloor }^{\left(m\right)}\approx\left(1-G_{\lambda}\right)^{m}n^{-1/2}x_{t=\lfloor nr\rfloor}\rightsquigarrow0,
\]
 thereby recovering the consistency of the trend component $n^{-1/2}\widehat{f}_{\left\lfloor nr\right\rfloor }^{\left(m\right)}\rightsquigarrow B(r)$, as there is no regular cycle component in the unit root model.\footnote{Unit root and many other nonstationary models generate stochastic trends. Such trends often produce trajectories with some form of highly irregular cyclical behavior because pure unit root processes, just as Brownian motion, return infinitely often to the origin; but such cycles are certainly not regular dynamic cycles of the type generated by stationary dynamic autoregressive processes or which were, at least in the past, associated with deterministic dynamic systems. Thinking about business cycles has inevitably evolved in recent years, taking into account historical experience where the period, intensity and irregularity of business cycles and recessions are frequently acknowledged and characterized in the popular parlance by which they are commonly distinguished, as discussed in \citetalias{phillips2021boosting}. Stochastic trends of very general forms may provide appropriate generative models that encompass a wide range of time series trajectories, inclusive of such irregular cycles, that are suited to economic series like unemployment, interest rates and inflation that are mildly nonstationary, as well as series like GDP that typically have some near random walk component.}

\subsection{Higher Order Integrated Processes\label{subsec:Case-1:-High}}

Many other types of nonstationary time series besides $I(1)$ processes occur in macroeconomic data. For instance, aggregate measures of the money supply and nominal price series are often well modeled by higher order integrated time series, particularly by $I(2)$ processes \citep{johansen1995stastistical,haldrup1998econometric}.
The KL series of the limiting Brownian motion process involves orthonormal series of sine functions
but it is equally clear from \eqref{eq:heur1} and \eqref{eq:heurstic} that similar shrinking
factors apply to cosine series and more general trigonometric polynomial functions. Such functions figure in series representations of higher order integrated processes.

Suppose the stochastic trend $f_{t}$ follows a higher order integrated process $I(q)$, for some integer $q\in\left\{ 2,3,\ldots\right\}$, of the form 
\begin{equation*}
\left(1-\mathbb{L}\right)^{q}f_{t}=u_{t}, 
\end{equation*}
where  $u_{t}$ is a linear process satisfying (\ref{eq:wold}). The repeated summation form of $f_t$ from initialization at $t=0$ gives $f_t = \sum_{j_q=1}^t \sum_{j_{q-1}=1}^{j_q} \cdots \sum_{j_1=1}^{j_2}u_{j_1} + p_{q-1}(t)$ where $p_{q-1}(t)$ is a polynomial in $t$ of degree $q-1$ with constant coefficients. Standard weak convergence methods lead to the following limit process of the observed time series after rescaling 
\begin{equation}
\dfrac{x_{t=\lfloor n\cdot\rfloor}}{n^{q-0.5}}\rightsquigarrow B_{q}\left(\cdot\right):=\int_{0}^{\cdot}\int_{0}^{s_{q-1}
}\text{\ensuremath{\cdots\int_{0}^{s_{3}}\int_{0}^{s_{2}}B(s_{1})ds_{1}ds_{2}\cdots ds_{q-2}ds_{q-1}}}, \label{eq:convergenceIq}
\end{equation}
as $n \to \infty$. Uniform convergence in (\ref{uniformI}) ensures that in the expanded probability space we have the corresponding result for $x_{t}/n^{q-0.5}$, viz.,
\begin{equation}
\sup_{0\leq t\leq n}\left|\dfrac{x_{t}}{n^{q-0.5}}-B_{q}\left(\frac{t}{n}\right)\right|=o_{a.s.}\left(1\right).\label{uniformIq}
\end{equation}
The orthonormal series representation of $B_{q}(r)$ is obtained by termwise integration in view of the uniform and almost sure convergence of the KL series for Brownian motion in \eqref{eq:KLI1}, giving    
\begin{align}
B_{q}(r) & =\sum\limits _{k=1}^{\infty}\xi_{k}\sqrt{\lambda_{k}}\text{\ensuremath{\int_{0}^{r}}}\int_{0}^{s_{q-1}}\text{\ensuremath{\cdots\int_{0}^{s_{3}}\int_{0}^{s_{2}}\varphi_{k}(s_{1})ds_{1}ds_{2}\cdots ds_{q-2}ds_{q-1}}}\nonumber \\
 & =\text{\ensuremath{\sqrt{2}}}\sum\limits _{k=1}^{\infty}\left\{ \sum_{\ell=1}^{\lfloor q/2\rfloor}(-1)^{\ell-1}\lambda_{k}^{\ell}\dfrac{r^{q-2\ell}}{(q-2\ell)!}+\lambda_{k}^{q/2}\text{Im}\left[\left(-\mathbf{i}\right)^{q-1}\mathrm{e}^{(\mathbf{i}/\sqrt{\lambda_{k}})r}\right]\right\} \xi_{k}.\label{eq:Lq}
\end{align}
The braces in \eqref{eq:Lq} include two terms: $\sum_{\ell=1}^{\lfloor q/2\rfloor}(-1)^{\ell-1}\lambda_{k}^{\ell}r^{q-2\ell}/(q-2\ell)!$
is a $\left(q-2\right)$th order polynomial, and $\lambda_{k}^{q/2}\text{Im}\left[\left(-\mathbf{i}\right)^{q-1}\mathrm{e}^{(\mathbf{i}/\sqrt{\lambda_{k}})r}\right]$
alternates between a sine (for odd $q$) and a cosine
(for even $q$). 
\begin{example}
Setting $q=2$ makes $x_t$ an $I(2)$ process.
Rescaling by $n^{3/2}$ and letting $n \to \infty$ we have 
\begin{gather}
X_{n}(\cdot)=\dfrac{x_{t=\lfloor n\cdot\rfloor}}{n^{3/2}}\rightsquigarrow B_{2}(\cdot)\equiv\int_{0}^{\cdot}B(s)ds, \; \textrm{ with}\label{convergenceI2}\\
B_{2}(r)=\sqrt{2}\sum\limits _{k=1}^{\infty}\dfrac{1-\cos\left[(k-\frac{1}{2})\pi r\right]}{[(k-\frac{1}{2})\pi]^{2}}\xi_{k}=\sum\limits _{k=1}^{\infty}\lambda_{k}[\sqrt{2}-\psi_{k}(r)]\xi_{k},\label{eq:L2}
\end{gather}
 where $\{\psi_{k}(r)=\sqrt{2}\cos\left[(k-\frac{1}{2})\pi r\right]=\sqrt{2}\cos\left[r/\sqrt{\lambda_{k}}\right]\}_{k=1}^{\infty}$
is an orthonormal system of cosine series. The series  \eqref{eq:L2}
of $B_{2}(r)$ converges faster than that of $B(r)$ since the decay rate of the coefficients $\lambda_k$ exceeds that of $\sqrt{\lambda_k}$ as $k \to \infty$. Correspondingly,
$B_{2}(r)$ is a smooth (once differentiable) Gaussian stochastic process in contrast to the Brownian motion $B(r).$ The first derivative of  $B_{2}(r)$ is $B(r)$, and $B_{2}(r)$ has a zero initial value at the origin.
\end{example}
It was long held as conventional wisdom that the HP filter removed up to four unit roots, thereby detrending integrated processes up to the fourth order \citep{king1993low}. 
This claim was disproved by \citetalias{phillips2021business} for $I(1)$ processes under the expansion
rate $\lambda=\mu n^{4}$, which was shown to match common quarterly time series applications in macroeconomics. The next result does the same for $I(2)$ processes, showing that the smoothing property of the HP filter continues to apply in this case as $n\to\infty$, giving a limit representation that is a smoothed version of $B_{2}(r)$ rather than $B_{2}\left(r\right)$ itself. 

\begin{prop}
\label{pop:hpI2} If $x_{t}$ satisfies the functional law (\ref{convergenceI2})
and $\lambda=\mu n^{4}$, then the HP filtered series has the following
limit form in the extended probability space as $n\to\infty$:
\begin{equation}
\dfrac{\hat{f}_{\lfloor nr\rfloor}^{\mathrm{HP}}}{n^{3/2}}\to_{a.s.}\sum_{k=1}^{\infty}\lambda_{k}\left[\sqrt{2}-\dfrac{\lambda_{k}^{2}}{\mu+\lambda_{k}^{2}}\psi_{k}\left(r\right)\right]\xi_{k}=:f^{\mathrm{HP}}(r).\label{eq:HPI2limit}
\end{equation}
\end{prop}
Under $\lambda=\mu n^{4},$ Proposition \ref{pop:hpI2} shows that
the HP filtered $I(2)$ trend approaches a limiting stochastic process $f^{\mathrm{HP}}(r)$
which deviates from the limiting trend process $B_{2}(r)$ and is
therefore inconsistent for this expansion rate of $\lambda.$ Correspondingly,
the estimated $\widehat{c}_{t}^{\text{HP}}$ of the cycle component
$c_{t}$ has the following limiting functional form 
\begin{equation}
\frac{\hat{c}_{\lfloor nr\rfloor}^{\text{HP}}}{n^{3/2}}
=n^{-3/2} \left( x_{\lfloor nr\rfloor} - \hat{f}_{\lfloor nr\rfloor}^{\text{HP}} \right)
\to_{a.s.}\sum_{k=1}^{\infty}\dfrac{-\mu\lambda_{k}}{\mu+\lambda_{k}^{2}}\psi_{k}\left(r\right)\xi_{k} =:c^{\text{HP}}(r)\label{eq:L2_c}
\end{equation}
upon standardization in the expanded space. This limit function is
a stochastic process that inherits some of the stochastic trend properties
of the limiting process $B_{2}(r).$ It is therefore to be expected
that with a smoothing parameter that approximates $\lambda=\mu n^{4},$
the HP filter will fail to remove all the trend properties
of the $I(2)$ process and the imputed business cycle estimate $\hat{c}_{t}^{\mathrm{HP}}$
will inevitably carry some of these `spurious' characteristics. Note also that at the origin the filtered series limit function is 
$f^{\mathrm{HP}}(0)=\sum_{k=1}^{\infty}\frac{\sqrt{2} \mu \lambda_{k}}{\mu+\lambda_{k}^{2}} \xi_{k}
= -c^{\text{HP}}(0) \not = 0$, 
a random, mean zero, initialization that is different from $B_2(0)=0$.

\begin{rem}
Given the fact $\lambda_{k}\asymp1/k^{2},$\footnote{For any two positive sequences $a_{n}$ and $b_{n}$, we use $a_{n}\asymp b_{n}$
to signify that $C^{-1}b_{n}\leq a_{n}\leq Cb_{n}$ for some positive
constant $C\in(1,\infty)$ as $n$ is sufficiently large. } the coefficients in the series representation (\ref{eq:HPI2limit}) satisfy 
$
\lambda_{k}^{3} / (\mu+\lambda_{k}^{2}) \asymp k^{-6},
$
from which we deduce that the limit process in (\ref{eq:HPI2limit})
is a Gaussian stochastic process that is continuously differentiable
to the fifth order, with fifth derivative 
\[
[f^{\mathrm{HP}}(r)]^{(5)}=\sum_{k=1}^{\infty}\left[\dfrac{\sqrt{\lambda_{k}}}{\mu+\lambda_{k}^{2}}\varphi_{k}\left(r\right)\right]\xi_{k},
\]
which is a non-differentiable Gaussian process for all $\mu>0$ similar to the Brownian motion $B(r)=\sum\limits _{k=1}^{\infty}\sqrt{\lambda_{k}}\varphi_{k}(r)\xi_{k}.$
Thus,  the trend extracted by the HP filter when $\lambda=\mu n^{4}$
is a very smooth function.

\bigskip
\end{rem}
The inconsistency of the HP filter estimate of $B_{2}(r)$ in \eqref{eq:HPI2limit}
is anticipated in view of \citetalias{phillips2021business}'s
earlier findings for HP filtering of an $I(1)$ stochastic trend. The next result shows that boosting restores consistency to the HP filter for  $I(2)$ and higher order integrated time series. For practical applications, the results for $I(2)$ time series are clearly the most relevant.  
\begin{thm}
\label{thmbhpIq} Suppose that $x_{t}$ satisfies the functional law
(\ref{eq:convergenceIq}). Given $\lambda=\mu n^{4}$, the bHP filter
has the following standardized limit theory  
\[
n^{0.5-q}\hat{f}_{\lfloor nr\rfloor}^{(m)}\rightsquigarrow B_{q}(r)
\]
for all positive integers $q\in \mathbb{N}$ as $m,n\to\infty$.
\end{thm}

\noindent When $q=1$ this result includes the unit root $I(1)$ case with $B_1=B$. 
The generalization for $q\geq 2$ follows by use of Lemma \ref{lem:m-operator} and the asymptotic representation of repeated applications of the HP operator on the exponential functions.

\subsection{Local Unit Root Processes \label{subsec:Local-Unit-Root}}

While models with unit roots provide a prototypical framework for capturing persistence in time series data, these models have modifications designed to capture a wider class of time series behavior in which the autoregressive roots are not restricted to unity as they are with integrated processes. An important subclass of more general models with near unit roots \citep{phillips2023estimation} is the class of LUR models 
\begin{equation} \label{eq:LURmodel}
(1-\mathrm{e}^{c/n}\mathbb{L})\, f_{t}=u_{t}, \; \:  t=1,2, \dots,n, \; \textrm{with } f_0=o_p(\sqrt{n}), 
\end{equation}
in which the autoregressive root ${\rm e}^{c/n} \approx 1+c/n$ is local to unity 
for some constant $c\in\mathbb{R}$ and large $n$. These models were developed in \citep{phillips1987towards,chan1987} and have been used for power analyses and in empirical research to provide robustness against pure unit root specifications. Time series generated by \eqref{eq:LURmodel} are nonstationary and, after suitable standardization, the observed time series $x_{t}$ converges to a linear diffusion, or Ornstein-Uhlenbeck (OU), process
\begin{equation}
X_{n}(r):=\dfrac{x_{\lfloor nr\rfloor}}{\sqrt{n}}=\int_{0}^{r}\mathrm{e}^{(r-s)c}dX_{n}(s)+O(n^{-1/2})\rightsquigarrow J_{c}(r):=\int_{0}^{r}\mathrm{e}^{(r-s)c}dB(s),\label{eq:convergenceltu}
\end{equation}
as $n \to \infty$, where $B(\cdot)$ is Brownian motion with variance $\omega^2$ as in \eqref{convergenceI1}.
When $c<0$ the limiting OU process $J_{c}$ is stationary and mean-reverting; when $c>0$
the process is explosive \citep{phillips1987towards,phillips2007limit}. 

Using Lemma 3.1 of \citet{phillips2007unit}, the uniform convergence
law (\ref{uniformI}) holds, ensuring that
\begin{equation}
\sup\limits _{0\leq t{}\leq n}\left|X_{n}\left(t/n\right)-J_{c}\left(t/n\right)\right|=o_{a.s.}(1)\label{uniformltu}
\end{equation}
 in the expanded probability space. A convenient series representation of $J_{c}(r)$ is 
\begin{equation}
\begin{aligned}J_{c}(r) & =\sum\limits _{k=1}^{\infty}\dfrac{1}{\lambda_{k}c^{2}+1}\left[\sqrt{\lambda_{k}}\varphi_{k}(r)+\sqrt{2}c\lambda_{k}\mathrm{e}^{cr}-c\lambda_{k}\psi_{k}(r)\right]\xi_{k},\label{eq:JcrKL}
\end{aligned}
\end{equation}
as derived by \eqref{eq:JCK-detail} in the Appendix.
Note that the first component of \eqref{eq:JcrKL} has the form
\begin{equation*}
\sum_{k=1}^{\infty}\frac{\sqrt{\lambda_{k}}}{\lambda_{k}c^{2}+1}\varphi_{k}(r)\xi_{k}
=\sum_{k=1}^{\infty}\left(\sqrt{\lambda_{k}}-\frac{\lambda_{k}^{3/2}c^2}{\lambda_{k}c^{2}+1}\right)\varphi_{k}(r)\xi_{k}	
\end{equation*}	
in which $\sum_{k=1}^{\infty}\sqrt{\lambda_{k}}\varphi_{k}(r)\xi_{k}=B(r)$, corresponding to the leading (non-differentiable) Brownian motion component in the decomposition $J_{c}(r) = B(r)+c\int_{0}^{r}\mathrm{e}^{(r-s)c}B(s)ds$. 
The remaining terms of \eqref{eq:JcrKL} provide a series representation of the smooth component $c\int_{0}^{r}\mathrm{e}^{(r-s)c}B(s)ds$ of $J_{c}(r)$, one of which is the exponential term
$\sqrt{2}c\nu\mathrm{e}^{cr}$ with a random Gaussian coefficient $\nu:=\sum_{k=1}^{\infty}\frac{\lambda_{k}}{1+c^2\lambda_k}\xi_{k} \sim 
N\left(0, \sigma^2_{\nu}\right)$,
where $\sigma^2_{\nu} = \omega^2\sum_{k=1}^{\infty}\left(\frac{\lambda_{k}}{1+c^2\lambda_k}\right)^2$.

\vspace{1.5mm}
The following proposition shows that the limit representation
of the HP filtered LUR time series is inconsistent, yielding a smoothed version of the diffusion $J_{c}(r).$ 
\begin{prop}
\label{pop:hpltu} If $x_{t}$ satisfies the functional law (\ref{eq:convergenceltu})
and $\lambda=\mu n^{4}$, then the HP filtered series has the following
limiting form as $n\to\infty$:
\begin{equation}
\dfrac{\hat{f}_{\lfloor nr\rfloor}^{\text{HP}}}{n^{1/2}}\to_{a.s.}\sum\limits _{k=1}^{\infty}\frac{1}{\lambda_{k}c^{2}+1}\left[\dfrac{\sqrt{2}c\lambda_{k}\mathrm{e}^{cr}}{\mu c^{4}+1}+\frac{\lambda_{k}^{2}}{\mu+\lambda_{k}^{2}}\left(\sqrt{\lambda_{k}}\varphi_{k}(r)-c\lambda_{k}\psi_{k}(r)\right)\right]\xi_{k}=:f_{\mathrm{LUR}}^{\mathrm{HP}}(r).\label{eq:HPLURlimit}
\end{equation}
\end{prop}

\bigskip

\begin{rem} 
Since $\lambda_{k}\asymp1/k^{2}$, the coefficients associated with
the sine and cosine waves in (\ref{eq:HPLURlimit}) satisfy
\[
\dfrac{\lambda_{k}^{5/2}}{(\lambda_{k}c^{2}+1)(\mu+\lambda_{k}^{2})}\asymp\dfrac{1}{k^{5}},\ \dfrac{c\lambda_{k}^{3}}{(\lambda_{k}c^{2}+1)(\mu+\lambda_{k}^{2})}\asymp\dfrac{1}{k^{6}},
\]
respectively. The real exponential function component $\sqrt{2}c\mathrm{e}^{cr}\sum\limits _{k=1}^{\infty}\frac{\lambda_{k}\xi_k}{(\lambda_{k}c^{2}+1)(\mu c^{4}+1)}$ has a random coefficient and
is infinitely differentiable. The limit process is therefore a Gaussian
stochastic process continuously differentiable to the fourth order
with the  fourth derivative given by 
\begin{align*}
[f_{\mathrm{LUR}}^{HP}(r)]^{(4)}= \sum\limits _{k=1}^{\infty}\frac{1}{\lambda_{k}c^{2}+1}\left[\frac{\sqrt{2}c^5 \lambda_{k}\mathrm{e}^{cr}}{\mu c^{4}+1}
+\dfrac{1}{\mu+\lambda_{k}^{2}}\left(\sqrt{\lambda_{k}}\varphi_{k}(r)-c\lambda_{k}\psi_{k}(r)\right)\right]\xi_{k},
\end{align*}
which is a convergent series.
\end{rem}

\medskip

\begin{rem} \label{rem:Surprise1}
The limits of the HP estimated trend \eqref{eq:HPLURlimit} and the HP estimated cycle 
\begin{equation}
\dfrac{\hat{c}_{\lfloor nr\rfloor}^{\mathrm{HP}}}{n^{1/2}}\to_{a.s.}\sum\limits _{k=1}^{\infty}\frac{1}{\lambda_{k}c^{2}+1}\left[
\frac{\sqrt{2} \mu  c^5 \lambda_{k}\mathrm{e}^{cr}}{\mu c^{4}+1} 
+\frac{\mu}{\mu+\lambda_{k}^{2}}\left(\sqrt{\lambda_{k}}\varphi_{k}(r)-c\lambda_{k}\psi_{k}(r)\right)\right]\xi_{k}
\label{eq:LUR_c}
\end{equation}
have an exponential trend component ${\rm e}^{cr}$, scaled respectively by the positively correlated zero mean Gaussian coefficients $\frac{\sqrt{2}c}{1+\mu c^{4}}\nu$ and $\frac{\sqrt{2}\mu c^{5}}{1+\mu c^{4}}\nu$
whose covariance is
$2\mu \omega^2 \frac{c^{6}}{\left(1+\mu c^{4}\right)^2}
\sigma^2_{\nu}>0$
for $\mu>0$. For $c<0$ the deterministic factor ${\rm{e}}^{cr}$ induces exponential decay in both the limiting HP fitted trend and fitted cycle. When $c>0$ the factor $\mathrm{e}^{cr}$ induces exponential growth in these components. 
When $c=0$ the limit in \eqref{eq:HPLURlimit} corresponds to
\begin{align*}
\dfrac{\hat{f}_{\lfloor nr\rfloor}^{\text{HP}}}{n^{1/2}} \to_{a.s.} \sum\limits _{k=1}^{\infty}\frac{\lambda_{k}^{2}}{\mu+\lambda_{k}^{2}}\sqrt{\lambda_{k}}\varphi_{k}(r)\xi_{k}, \;\;
\dfrac{\hat{c}_{\lfloor nr\rfloor}^{\text{HP}}}{n^{1/2}} \to_{a.s.} \sum\limits _{k=1}^{\infty}\frac{\mu}{\mu+\lambda_{k}^{2}}\sqrt{\lambda_{k}}\varphi_{k}(r)\xi_{k}, 
\end{align*}
giving the same findings as in \citetalias{phillips2021business}.
\end{rem}

\medskip

\begin{rem}
\label{rem:Surprise2} These latter two properties of the limiting residual process in (\ref{eq:LUR_c})
contrast with those in \eqref{eq:L2_c} for the HP fitted residual of $B_{2}(r)$,
where the HP filter removes the polynomial (in this case, intercept) component in the residual, leaving only the trigonometric functions. 
The HP filter's elimination of the intercept of $B_{2}(r)$ in the fitted residual is explained by the fact that the HP filter removes time polynomial functions up to the third degree, thereby including the case of the intercept in the representation of $B_{2}$.
This facility does not include the exponential function
$\mathrm{e}^{cr}=\sum_{j=0}^{\infty}\left(cr\right)^{j}/j!$, which exceeds the capacity 
of the HP filter. Nonetheless, when $c<0$ the exponential decay factor ${\rm e}^{cr}$ diminishes the magnitude of this component in the residual.
Simulation evidence given in Appendix \ref{sec:Additional-Numerical-Results} corroborates the sign effects of $c$ on the estimation error of the HP filter.
\end{rem}

\bigskip

Whereas the HP filter fails to fully capture an exponential trend function, this objective
is fulfilled by the bHP filter, as confirmed in the next result.
\begin{thm}
\label{thmbhpltu} If $x_{t}$ satisfies the functional law (\ref{eq:convergenceltu})
and $\lambda=\mu n^{4}$, then the bHP filter is consistent with  
$n^{-1/2}\widehat{f}_{\lfloor nr\rfloor}^{(m)}\rightsquigarrow J_{c}(r)$
as $m,n\to\infty$.
\end{thm}
For any finite $c\in\mathbb{R},$ Theorem \ref{thmbhpltu} shows that boosting the HP
filter removes stochastic trend components in the residual and provides consistent recovery of a local to unity trend. 
With respect to the arguments given above in Remark \ref{rem:Surprise2} regarding the effects of filtering on a finite degree time polynomial trend, the limit theory of the bHP filter in
\citetalias{phillips2021boosting} (Theorem 2, p.~555) shows that 
the bHP filter in $m$ iterations
removes a polynomial trend of degree $\left(4m-1\right)$ from the residual cycle. 
Passing $m$ to infinity ensures that boosting captures all the terms in a power series expansion of an exponential trend, corroborating the consistency of the bHP filter given in Theorem \ref{thmbhpltu} here. 

\subsection{Linear Combination of Stochastic and Deterministic Trends}

Time series models sometimes explicitly include deterministic trends and structural breaks as constituent
trend components. Such components may be analyzed for the wider class of stochastic trend functions considered in this paper. 
In particular, following \citetalias{phillips2021boosting} we consider a deterministic trend given by the time polynomial
\begin{equation}
d_{n}(t)=\alpha_{n}+\beta_{n,1}t+\dots+\beta_{n,J}t^{J}.\label{eq:deter}
\end{equation}
It is a direct corollary that bHP consistently estimates the trend
if the underlying stochastic trend
is accompanied by an additive polynomial deterministic trend.
\begin{cor}
\label{cor:bhpIqdeter} Suppose that $x_{t}^{0}$ satisfies the functional
law $\dfrac{x_{\lfloor nr\rfloor}^{0}}{n^{q-0.5}}\rightsquigarrow B_{q}(r)$.{}
Let $x_{t}=x_{t}^{0}+d_{n}(t)$ where $d_{n}(t)$ is given by (\ref{eq:deter})
and the coefficients satisfy $\alpha_{n}/n^{q-0.5}\to\alpha$ and
$n^{j-q+0.5}\beta_{n,j}\to\beta_{j}$ for $j=1,2,\dots,J$. Given
$\lambda=\mu n^{4}$, the asymptotic form of the bHP filtered trend is  
\begin{equation*}
\dfrac{\hat{f}_{\lfloor nr\rfloor}^{(m)}}{n^{q-0.5}}\rightsquigarrow B_{q}(r)+d(r),
\end{equation*}
as $m,n\to\infty$, where $d(r)=\alpha+\beta_{1}r+\dots+\beta_{J}r^{J}.$
\end{cor}

A parallel result  holds giving consistency of the bHP filter for an LUR time series with deterministic time polynomial drifts and structural breaks. Simply set $q=1$ and replace $B_{q}\left(r\right)$
by $J_{c}\left(r\right)$ in Corollary \ref{cor:bhpIqdeter}. 
For any stochastic trend considered in this paper, 
if the deterministic trend is consists of finite piece-wise polynomials, 
then consistent estimation can be obtained following 
\citetalias{phillips2021boosting}'s Theorem 3. 
Full statements are omitted.


\bigskip 

Given the consistency with the deterministic trends, it is easy to see that finitely additive normalized combinations of these trends are correspondingly included. For example, suppose 
\[
x_{t}=\frac{x_{t}^{\left(1\right)}}{n}+x_{t}^{\left(2\right)}+\sqrt{n}\alpha+\frac{\beta}{\sqrt{n}}t
\]
where $x_{t}^{\left(1\right)}$ is an $I(2)$ trend as in (\ref{convergenceI2}),
$x_{t}^{\left(2\right)}$ is an LUR trend as in (\ref{eq:convergenceltu}), and the respective innovations of these time series $(u_{t}^{(1)},u_{t}^{(2)})'$
are potentially correlated random variables satisfying a bivariate
version of (\ref{eq:wold}). Then the boosted HP filter reproduces the asymptotic form of this combined trend process so that
\[
n^{-1/2}\widehat{f}_{\lfloor nr\rfloor}^{(m)}\rightsquigarrow B_{2}(r)+J_{c}(r)+\left(\alpha+\beta r\right)
\]
given $\lambda=\mu n^{4}$ as $m,n\to\infty$. Such results demonstrate the versatility of the boosted filter in dealing with complex forms of nonstationary time series. 

\section{Simulations\label{sec:Simulations}}

Practical implementation of boosting requires two tuning parameters $\lambda$
and $m$ to be selected. For low frequency macroeconomic data, \citetalias{phillips2021boosting} 
suggested that the conventional choice $\lambda=1600$ be used for quarterly
data, so that the first iteration gives the HP filter. This choice provides a
benchmark that can be adjusted in the case of annual or monthly data \citep{ravn2002adjusting}.
To determine the boosting number $m$ \citetalias{phillips2021boosting} proposed a data-driven stopping rule based on a Bayesian-type information criterion (BIC)
\begin{equation}
IC(m)=\dfrac{\text{\ensuremath{\widehat{c}^{(m)\prime}\widehat{c}^{(m)}}}}{\widehat{c}^{\mathrm{HP}\prime}\widehat{c}^{\mathrm{HP}}}+\log n\dfrac{\text{\text{tr}}\left[I_{n}-(I_{n}-S)^{m}\right]}{\text{tr}(I_{n}-S)},\label{eq:BIC}
\end{equation}
which is motivated by a bias-variance trade-off.\footnote{In addition to the BIC \eqref{eq:BIC},
\citetalias{phillips2021boosting} also suggested a stopping rule based on an augmented Dickey-Fuller 
unit root test, according to the notion that cyclical
behavior should not exhibit unit root behavior. Simulations in \citetalias{phillips2021boosting} 
showed that \eqref{eq:BIC} generally provided better finite sample performance and has therefore been used here.
} 
Unless explicitly stated otherwise, 
the numerical work of bHP in this paper employs the following algorithm.

\begin{lyxalgorithm}
\label{alg:Automated-bHP-BIC} The boosted HP filter (with the BIC stopping rule)
\end{lyxalgorithm}
\begin{description}
\item [{Step1}] Set a maximum number of iterations, say, $m_{\mathrm{max}}=200$.
Specify the smoothing parameter $\lambda=1600$ for quarterly data,
$\lambda=129600$ for monthly data, or $\lambda=6.25$ for annual
data. 
\item [{Step2}] Run Algorithm \ref{alg:Boosted-HP-filter} with $m=m_{\mathrm{max}}$,
and compute $IC\left(m\right)$ in (\ref{eq:BIC}) at each iteration. 
\item [{Step3}] Given the minimizer $\widehat{m}:=\min_{m\leq m_{\mathrm{max}}}IC\left(m\right)$,
compute $\widehat{f}^{\left(\widehat{m}\right)}$ and $\widehat{c}^{\left(\widehat{m}\right)}$. 
\end{description}
Code to implement the above algorithm is provided. In \texttt{R} software simply call
\begin{quotation}
\texttt{BoostedHP(x, lambda = 1600, stopping = \textquotedbl BIC\textquotedbl ,
Max\_Iter = 200)}\footnote{An \texttt{R} package can be accessed at \url{https://github.com/zhentaoshi/bHP_R_pkg}
\citep{chenshi2021}, and this command line is the default setting of the
arguments in the function. Parallel open-source functions are also
available in \texttt{Matlab} and \texttt{Python} at \url{https://github.com/zhentaoshi/Boosted_HP_filter}.}
\end{quotation}
where \texttt{x} is the observed time series, and this function will
return the estimated trend and cyclical components, along with the
sequence of values of $IC\left(m\right)$ and the number of iterations
$\widehat{m}$.


\subsection{Data Generating Processes}

In the following numerical exercises, we generate time series in the general form of (\ref{eq:y_decomp}).
Along with various stochastic and deterministic trend processes,
we specify a cyclical component $c_t$ satisfying a stationary AR(2)\footnote{
We have experimented with many designs of the stationary component and found that the relative performance of the filters is robust. Due to limitations of space we report here the results of a single illustrative design.}
$$
(1 - \beta_1 \mathbb{L} - \beta_2 \mathbb{L}^2)\,  c_{t}=  e_{t},
\ \ \mbox{where } e_t \sim i.i.d.~N(0,\sigma^2_e). 
$$  
To mimic two-year business cycles, we fix $\beta_1 = 1$ and set $\beta_2 = -0.5469$ for quarterly data
to produce a spectral density function\footnote{See \citet[Property 4.3]{shumway2000time} for the formula of spectral density functions for ARMA processes.} peaking at $\omega = 1/8$ corresponding to a periodicity of $8$ quarters. Similarly, we set $\beta_2 = -0.3492$ for monthly data so that the spectral density peaks at $1/24$, reflecting a periodicity of $24$ months. 

We first design three data generating processes (DGPs) based on an $I(2)$ trend. DGP1 is a simple $I(2)$ process,  upon which DGP2 adds another 3rd order polynomial trend, and DGP3 further involves a structural break in the middle. All these cases are covered by the asymptotics in Section \ref{sec:theo-Applications}.  
\begin{lyxlist}{00.00.0000}
\item [{DGP1.}] $f_{t}^{(1)}$ satisfies $(1 - \mathbb{L})^2 f_{t}^{(1)}= v_{t}$, where $v_{t}\sim i.i.d.~N(0,1)$.
\item [{DGP2.}] $f_{t}^{(2)}=f_{t}^{(1)}+200\left(t/n\right)^{3}.$ A third degree polynomial trend is added to DGP1.
\item [{DGP3.}] $f_{t}^{(3)}=f_{t}^{(1)}+200\left(t/n\right)^{3}\cdot1\{t>0.5n\}$. A third degree polynomial trend with a structural break is added to DGP1.
\end{lyxlist}
While we set 
$\sigma_e=5$ for DGPs 1--3 so that the stationary component $c_{t}$ is non-negligible in finite samples in relation to the $I(2)$ trend,
we reduce $\sigma_e$ to 1 for the LUR trend 
to match the settings in \citetalias{phillips2021boosting} where unit root trends were employed, since 
LUR processes have less prominent trend behavior than $I(2)$ processes.
The following designs in DGPs 4--6 are used with a baseline LUR model parallel to DGPs 1--3. 
LURs from three groups (near explosive, unit root, near stationary) are used with $c \in \{3,0,-3\}$.

\begin{lyxlist}{00.00.0000}
\item [{DGP4.}] $f_{t}^{(6)}=\mathrm{e}^{c/n}f_{t-1}^{(6)}+v_{t}$.
\item [{DGP5.}] $f^{(7)}=f^{(6)}(t)+200\left(t/n\right)^{3}.$
\item [{DGP6.}] $f_{t}^{(8)}=f^{(6)}(t)+200\left(t/n\right)^{3}\cdot1\{t>0.5n\}.$
\end{lyxlist}

We use ``quarterly data'' and ``monthly data'' to refer to the respective sample sizes mimicking the FRED databases. Following Algorithm \ref{alg:Automated-bHP-BIC}, we use $\lambda=1600$ for quarterly data. Sample sizes are $n=100$ (25 years)
as a baseline, while $n=200$ (50 years) and $300$ (75 years) are comparable in size to the empirical application, where most time series in the FRED-QD database have 258 quarters. For monthly data $\lambda=129600$ is used and sample sizes are scaled up by a factor of three, giving $n \in \{300, 600, 900\}$ which is in line with the FREQ-MD database of 775 months.

\subsection{Simulation Results}

\begin{table}[t]
\centering 
\caption{MSE of the Estimated Trend: $I(2)$}\label{tab:I2}
\small
\begin{tabular}{ccrrrr|ccrrrr}
\hline\hline 
\multicolumn{6}{c}{Quarterly data}                 & \multicolumn{6}{c}{Monthly data}                                    \\
\hline  
DGP                & $n$   &  HP & 2HP & bHP &  HRF  & DGP                & $n$   &  HP & 2HP & bHP &  HRF \\
\hline 
\multirow{3}{*}{1} & 100 & 26.66 & 15.99 & 12.92 & 438.38 & \multirow{3}{*}{1} & 300 & 600.84 & 285.51 & 54.33 & 4376.96 \\
                   & 200 & 26.46 & 15.67 & 13.31 & 717.81 &                    & 600 & 608.18 & 288.03 & 59.83 & 5780.87 \\
                   & 300 & 26.26 & 15.50 & 13.67 & 885.12 &                    & 900 & 606.74 & 287.26 & 62.75 & 6272.95 \\
                   \hline 
\multirow{3}{*}{2} & 100 & 27.15 & 16.11 & 12.96 & 462.42 & \multirow{3}{*}{2} & 300 & 600.95 & 285.51 & 54.34 & 4374.37 \\
                   & 200 & 26.48 & 15.68 & 13.31 & 719.72 &                    & 600 & 608.19 & 288.03 & 59.85 & 5781.40 \\
                   & 300 & 26.26 & 15.50 & 13.67 & 885.71 &                    & 900 & 606.75 & 287.26 & 62.75 & 6273.40 \\
                   \hline 
\multirow{3}{*}{3} & 100 & 39.68 & 26.05 & 20.66 & 530.78 & \multirow{3}{*}{3} & 300 & 613.67 & 295.66 & 59.12 & 4493.40 \\
                   & 200 & 32.04 & 20.17 & 17.06 & 759.30 &                    & 600 & 614.22 & 292.71 & 62.18 & 5854.74 \\
                   & 300 & 29.99 & 18.49 & 16.23 & 917.18 &                    & 900 & 610.23 & 290.13 & 64.28 & 6323.76 \\
                   \hline\hline                    
\end{tabular}
\end{table}

The original HP filter, `twicing' $m=2$ iterated HP filter (2HP) \citep{hall2024selecting}, bHP
and the autoregressive method are applied to trended time series based on $I(2)$ and LUR. 
For each replication, the mean squared error
(MSE) is calculated as the sample average of $(\widehat{f}_{t}-f_{t})^{2}$, which is equivalent to the average of $(\widehat{c}_{t}-c_{t})^{2} = \left((x_{t}-\hat{c}_{t})-(x_{t}-c_{t})\right)^{2} = (f_{t} - \widehat{f}_{t})^{2} $.
That is, the square of the fitting errors of the trend and the cycle are identical.

We report the empirical MSE averaged over 5000 replications.
Table \ref{tab:I2} displays the MSEs for DGPs 1--3 where the stochastic
trends are based on $I(2)$ time series with additional trend and break components. It is evident that 2HP improves the original HP filter, and the bHP obtains the lowest MSEs. The MSEs for bHP for each sample size are all close across DGPs 1--2, suggesting that the additional deterministic trend components are
estimated accurately and the estimation errors stem primarily from the $I(2)$ stochastic trend.  
When the tuning parameter $\lambda$ rises from $1600$ for the quarterly data to $129600$ for monthly data, the MSEs for HP and 2HP are inflated in many cases tenfold or more. Evidently, the conventional choice $\lambda=129600$ for monthly data 
seems much too large for HP trend detection in monthly $I(2)$ time series. By contrast
the same large tuning parameter has a considerably muted effect on the trend detection performance of bHP.

For the HRF, $(h,p)=(8,4)$ is set for quarterly data and $(h,p)=(24,12)$ for
monthly data following \citet{hamilton2018you}. 
The MSEs from HRF are far larger than those from the HP filters appear increasing substantially across sample sizes, revealing the limitations of regression specifications as a flexible trend detection device.

\begin{table}[htbp]
\centering 
\caption{MSE of the Estimated Trends: LUR}\label{tab:LUR}
\small

\begin{tabular}{cc|rrrr|rrrr|rrrr}
\hline\hline
 & & \multicolumn{4}{c}{$c=3$} & \multicolumn{4}{c}{$c=0$} & \multicolumn{4}{c}{$c=-3$} \\ 
\hline
\multicolumn{14}{c}{Quarterly data} \\   
{DGP} &  {$n$}  & HP & 2HP    & bHP   & HRF  & HP & 2HP     & bHP   & HRF  & HP & 2HP     & bHP    & HRF  \\
                     \hline
                
\multirow{3}{*}{4} & 100 & 4.11  & 2.12  & 1.51 & 9.86  & 1.77  & 1.51  & 1.37 & 6.42  & 1.80  & 1.54  & 1.40 & 5.73  \\
                   & 200 & 1.80  & 1.50  & 1.43 & 9.53  & 1.80  & 1.52  & 1.46 & 7.76  & 1.82  & 1.54  & 1.48 & 7.23  \\
                   & 300 & 1.78  & 1.50  & 1.49 & 9.47  & 1.80  & 1.52  & 1.51 & 8.25  & 1.82  & 1.53  & 1.53 & 7.91  \\
                   \hline 
\multirow{3}{*}{5} & 100 & 4.58  & 2.23  & 1.55 & 16.09 & 2.26  & 1.62  & 1.37 & 17.87 & 2.28  & 1.65  & 1.41 & 17.73 \\
                   & 200 & 1.82  & 1.50  & 1.43 & 11.33 & 1.82  & 1.52  & 1.46 & 12.11 & 1.83  & 1.54  & 1.48 & 11.88 \\
                   & 300 & 1.78  & 1.50  & 1.49 & 10.29 & 1.80  & 1.52  & 1.51 & 10.80 & 1.82  & 1.53  & 1.53 & 10.72 \\
                   \hline 
\multirow{3}{*}{6} & 100 & 16.94 & 12.06 & 8.37 & 87.73 & 14.62 & 11.44 & 8.04 & 94.59 & 14.60 & 11.44 & 8.08 & 94.25 \\
                   & 200 & 7.49  & 6.04  & 4.71 & 41.51 & 7.48  & 6.06  & 4.73 & 43.47 & 7.48  & 6.06  & 4.74 & 43.17 \\
                   & 300 & 5.45  & 4.44  & 3.72 & 29.20 & 5.49  & 4.47  & 3.73 & 30.43 & 5.48  & 4.47  & 3.74 & 30.23 \\
                     \hline

\multicolumn{14}{c}{Monthly data} \\    
{DGP} &  {$n$}  & HP  & 2HP    & bHP   & HRF  & HP & 2HP     & bHP   & HRF  & HP  & 2HP    & bHP    & HRF  \\

                     \hline
\multirow{3}{*}{4} & 300 & 5.86  & 4.28  & 2.17 & 25.18 & 4.93  & 4.05  & 2.11 & 17.31  & 4.97  & 4.10  & 2.13 & 15.47  \\
                   & 600 & 5.06  & 4.12  & 2.21 & 25.40 & 5.06  & 4.13  & 2.22 & 21.15  & 5.05  & 4.13  & 2.22 & 19.76  \\
                   & 900 & 5.06  & 4.13  & 2.27 & 25.17 & 5.09  & 4.16  & 2.28 & 22.64  & 5.11  & 4.17  & 2.29 & 21.73  \\
                   \hline 
\multirow{3}{*}{5} & 300 & 6.37  & 4.40  & 2.22 & 31.32 & 5.42  & 4.17  & 2.16 & 34.71  & 5.46  & 4.21  & 2.18 & 33.85  \\
                   & 600 & 5.08  & 4.12  & 2.21 & 27.15 & 5.08  & 4.14  & 2.22 & 29.07  & 5.07  & 4.13  & 2.23 & 28.13  \\
                   & 900 & 5.06  & 4.13  & 2.27 & 26.09 & 5.09  & 4.16  & 2.28 & 27.70  & 5.12  & 4.17  & 2.29 & 27.15  \\
                   \hline 
\multirow{3}{*}{6} & 300 & 18.33 & 13.92 & 6.25 & 98.53 & 17.49 & 13.77 & 6.16 & 108.23 & 17.43 & 13.73 & 6.18 & 107.25 \\
                   & 600 & 10.69 & 8.62  & 4.18 & 56.23 & 10.68 & 8.66  & 4.20 & 59.40  & 10.65 & 8.62  & 4.18 & 58.47  \\
                   & 900 & 8.71  & 7.05  & 3.60 & 44.44 & 8.76  & 7.10  & 3.60 & 46.92  & 8.75  & 7.09  & 3.62 & 46.07\\
                    \hline\hline
\end{tabular}
\end{table}

Similar phenomena and rankings of MSEs are observed in Table \ref{tab:LUR}
for DGPs 4--6 when the underlying stochastic trend is LUR. When $n=100$
for quarterly data and $n=300$ for monthly data, the near-explosive
case with $c=3$ yields larger MSEs for HP. Under other sample sizes
the results are stable across the values $c \in \{3,0,-3\}$. The MSEs from bHP
are only slightly affected when $\lambda$ is raised from 1600 to
$129600$. Again, the bHP filter is uniformly superior to the competitors.

Regarding the number of iterations, under the $I(2)$ trends the average $m$ for bHP ranges from 2.5 to 4.2 for quarterly data, and from 15 to 20 for monthly data. This result is consistent with our theoretical result that a larger $\lambda$ in the HP filter for monthly data leads to an overly  smooth estimated trend and thus requires more iterations for consistency. Under the LUR trends the two ranges become 2.2--8 and 30--60.

In summary, the simulation results show that, in terms of MSEs the conventional
choices of $\lambda$ for quarterly and monthly data appear
too large for good trend detection by the HP filter. In all cases, HRF delivers worse MSEs in trend detection than HP. 
It is clear from these results that the iterated procedure makes the initial choice of the penalty parameter $\lambda$ less critical for performance, 
compensating for the shortcomings of HP particularly with more complex trend processes, thereby providing a more robust approach to general trend estimation than the other methods.
Results in \citetalias{phillips2021business} did show that when the penalty $\lambda=o(n)$, the HP filter is consistent for a stochastic trend when $n \to \infty,$ but no specific rate was suggested or explored. This finding is therefore not particularly helpful for empirical practice. But coupled with the consistency of bHP when $\lambda=\mu n^4$ and the robustness of bHP to choice of the penalty $\lambda$ in initiating iterations, it suggests the existence of a linkage between an `optimal' penalty $\lambda$ and the boosting iteration number $m$ that assures consistent trend estimation. Possible linkages of this type are being considered in ongoing work by \cite{yamada2023} and \cite{biswas2022longrange}.\footnote{\cite{biswas2022longrange} show that if $\lambda=c\mu n^4$ and $\hat \lambda = \hat c(m) \mu n^4$ is chosen to minimize the $L_2$ norm between the HP filter and the bHP filter after $m$ iterations, then the resulting HP with penalty $\hat \lambda$ is consistent when $ \hat c(m) \to 0$ as $m,n \to \infty$. \cite{yamada2023} derives a relationship between the smoothing matrices involved in the HP and bHP filters whose properties confirm the finding of \citetalias{phillips2021boosting} that `the asymptotic effect of increasing $m$ is similar to reducing
	the value of $\lambda$ in the simple HP filter'.}

\section{Empirical Applications}\label{sec:emp}
\subsection{Illustration: US GDP \label{sec:Appetizer}}

\begin{sidewaysfigure}
\begin{minipage}[t]{0.45\columnwidth}%
\subfloat[Estimated Trends]{\includegraphics[scale=0.32]{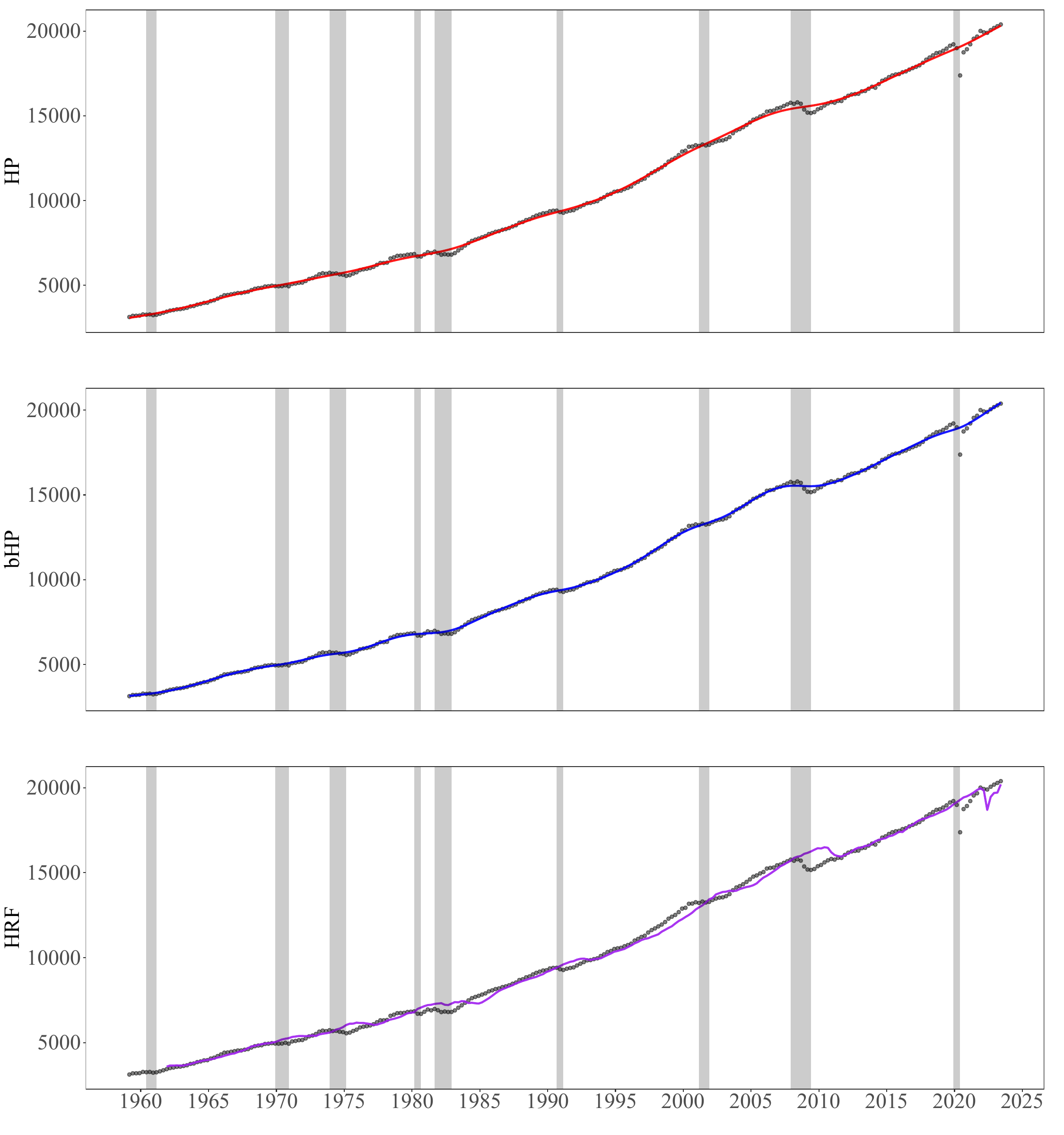}
}%
\end{minipage}\hfill{}%
\begin{minipage}[t]{0.45\columnwidth}%
\subfloat[Estimated Cycles]{\centering

\includegraphics[scale=0.32]{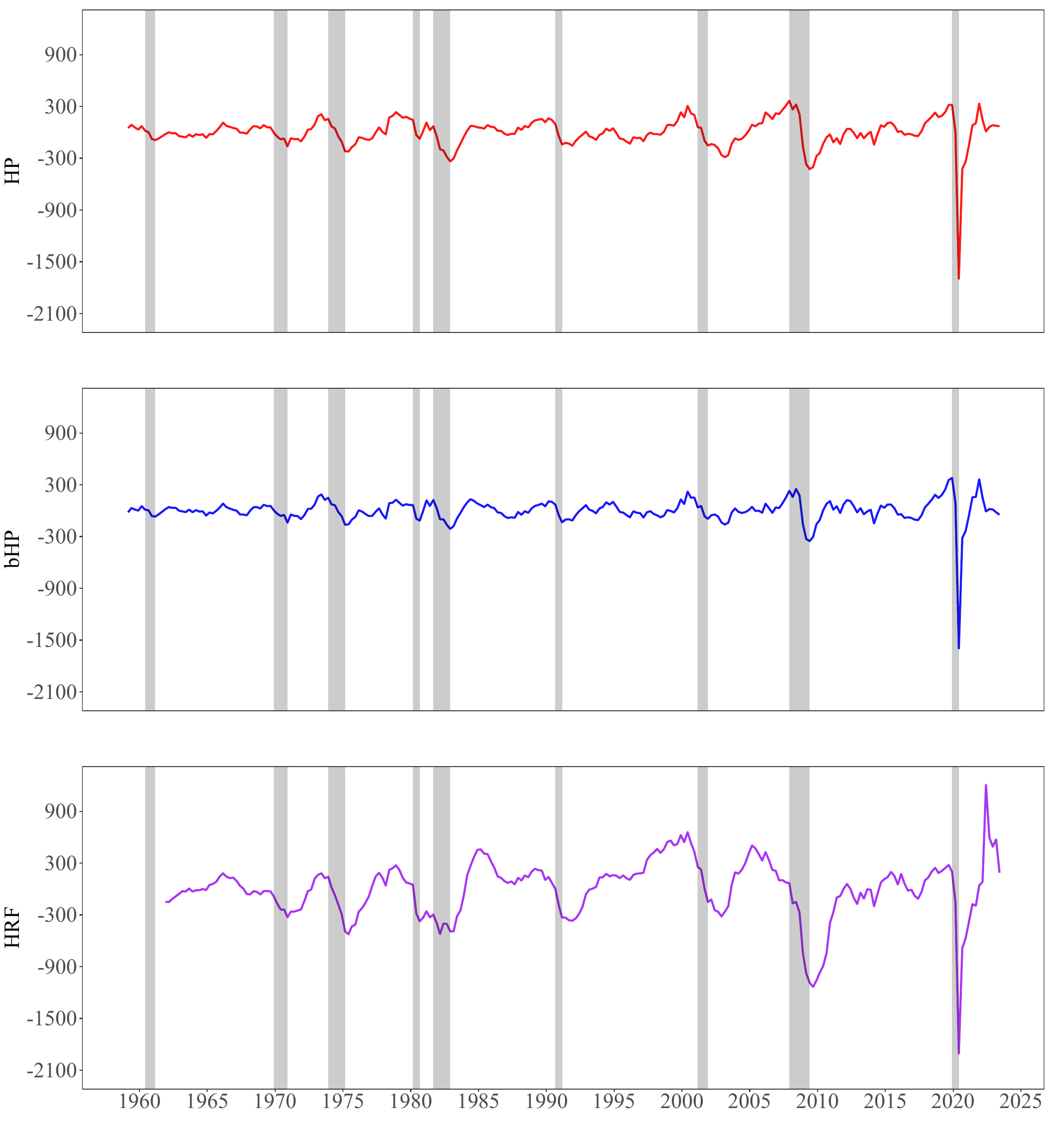}

}%
\end{minipage}
\caption{\label{fig:GDP-q} Trends, Cycles, and Raw Data of US Real GDP}
\end{sidewaysfigure}

\begin{sidewaysfigure}
\begin{minipage}[t]{0.45\columnwidth}%
\subfloat[Estimated Trends]{\includegraphics[scale=0.32]{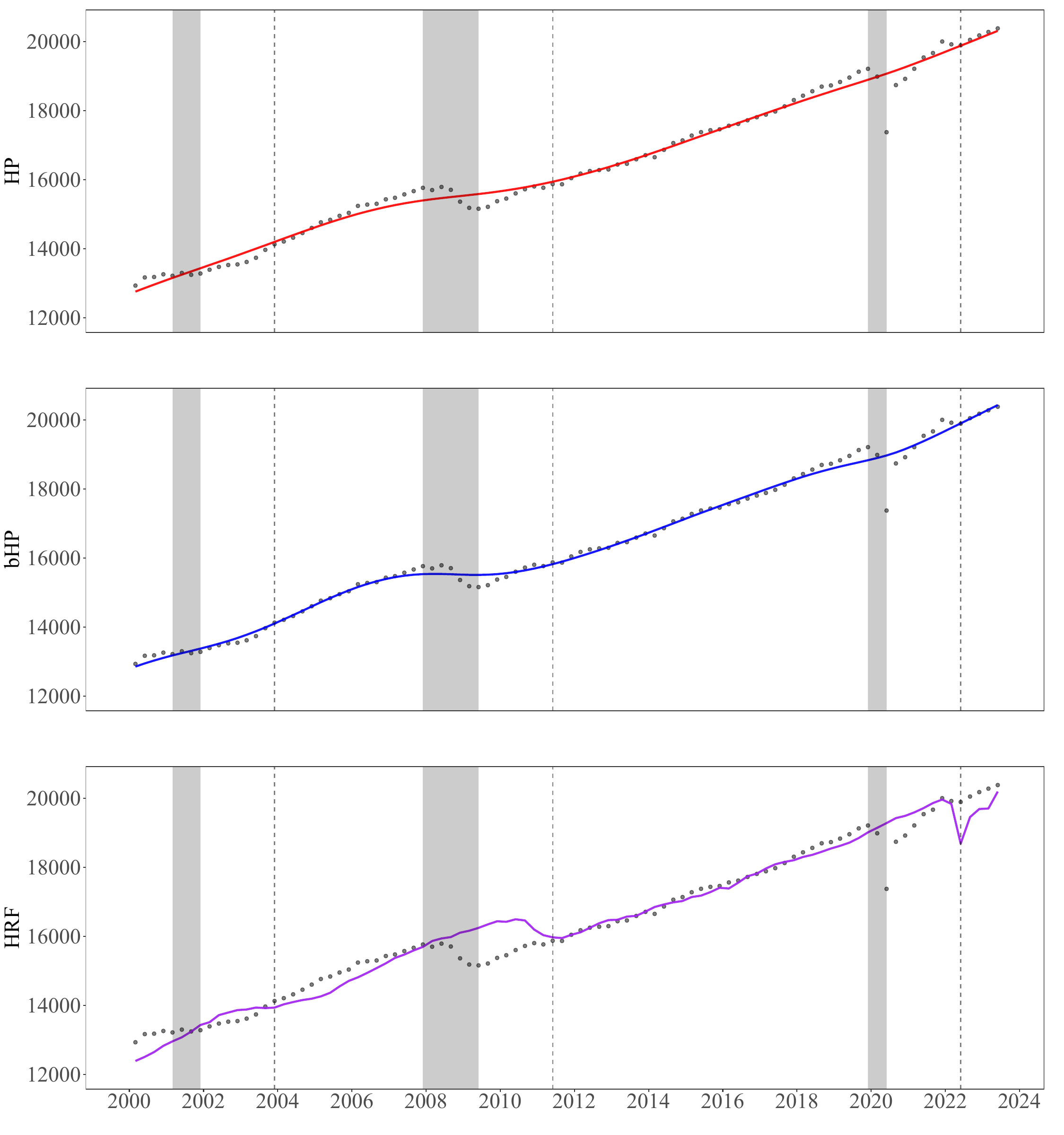}
}%
\end{minipage}\hfill{}%
\begin{minipage}[t]{0.45\columnwidth}%
\subfloat[Estimated Cycles]{\centering
\includegraphics[scale=0.32]{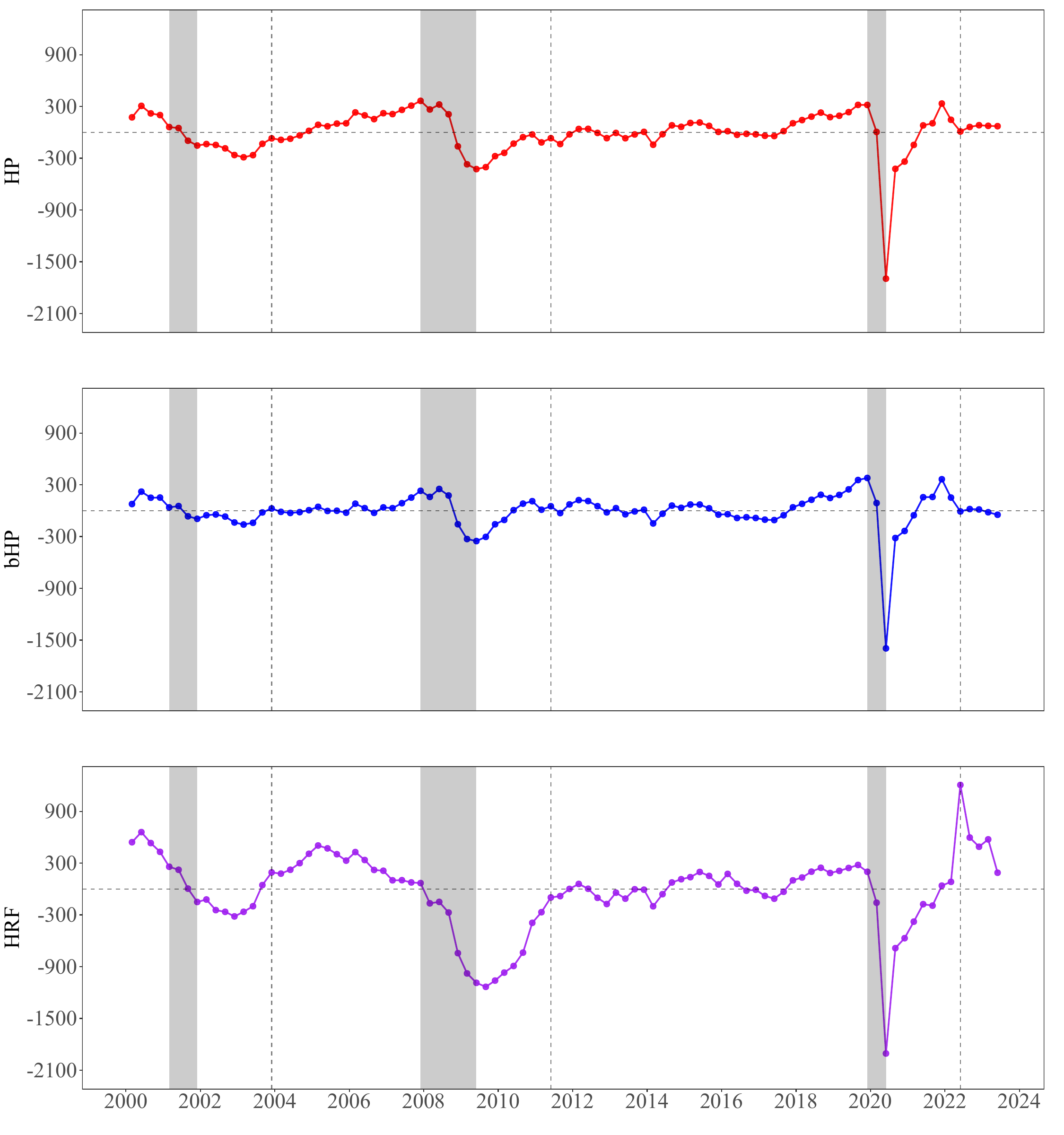}
}%
\end{minipage}

\footnotesize{Notes: In the right panels, the horizontal dotted line is $y=0$. The vertical dotted lines mark two-years after the last quarter of the recessions identified by NBER.}

\caption{\label{fig:GDP-q-21} Zoomed-in version of Figure \ref{fig:GDP-q}
after 2000}
\end{sidewaysfigure}


At the time
of writing, US real quarterly GDP in the FRED-QD
database consists of 258 quarterly observations from 1959:Q1 to 2023:Q2.
In Figure \ref{fig:GDP-q} the shaded time periods mark recessions
identified by the National Bureau of Economic Research (NBER). 

In the left panels,\footnote{Figures \ref{fig:GDP-q} and \ref{fig:GDP-q-21} are left-rotated $90^{\circ}$ to landscape mode in the paper. So the left (right) panel of the paper appears in the lower (upper) position in the rotated display. Similarly, the top, center and bottom panels refer to the unrotated graphics, whereas they are shown in $90^{\circ}$ rotated view in the paper. The descriptions in the text refer to the original unrotated orientation of the graphics.} the black dots in the graphics display the raw data, and the colored solid
lines are the trends as estimated by the HP filter (red line, top), bHP (blue line, middle), and HRF (violet line, lower).\footnote{The 2HP `twicing' filter was also calculated for all empirical examples but the results are not displayed as they are very close to those of HP.}
The HP filter produces a smooth, monotonically increasing trend. bHP is more responsive with a slight dip during the long
recession in 2008--09, leaving the residual to capture that long cycle; and the trends from HP do not decline in the short Covid-19 recession despite the enormous single period drop in the raw data in 2020:Q2, again leaving this dynamic to the cycle. 
By contrast, with $h=8$ the HRF largely shifts the observed data to the right by two years.

More revealing are the estimated cycles shown in the right panels, which measure the
gap between the raw data and the estimated trend. 
In the zoomed-in sub-figures, each data point is discernible.  The residuals from HP and bHP exhibit clear business cycles that capture the trough in all
recessions. The GFC was the most severe prior to 2020, but it was dwarfed in magnitude by the Covid-19 recession. 
The cyclical component took 1 year to return to positive after Covid-19.

The estimated cycle by HRF in Figure \ref{fig:GDP-q-21} shows some peculiarity. 
In all of the 21st-century three recessions, it take 2 years to recover, as marked by the dotted vertical lines. In particular, we observed an unprecedented peak in 2022:Q2, which is the consequence of the colossal slump in 2020:Q2 entering the 8-period ahead predictive model for the first time. 
These reversions are symptomatic of compensating over-prediction that can occur with autoregressive data fits,\footnote{See Figure \ref{fig:HRF-h} for $h = 1, 4, 8$ and $12$ for this built-in feature of HRF.}
which in turn accentuate the swings in actual economic activity.



\subsection{Full FRED Databases\label{sec:Empirical-Application}}

This section gives a `big
data' application, making full use of the time series in these 
two databases to draw an overall picture of trend and growth cycles in the US economy.
There are 246 quarterly series in the FRED-QD database and 127 series in the FREQ-MD database. This broad range of time series includes a mixture of trend behaviors. 

The quarterly data findings are presented first. FRED-QD is a comprehensive collection
of macroeconomic time series, covering 14 categories of economic activity concerning, for example,
the national income and product accounts, the labor market covering unemployment,
earnings and productivity, and financial markets covering interest rates, money and credit.
Most of the 246 individual time series are complete, with 258 quarters
spanning 1959:Q1 to 2023:Q2. A few sequences have missing values and the shortest one has 121 quarters. Many of the time series show strong trend behavior, such as the real GDP, the consumer price index, and stock market indices such as the S\&P 500. Other series have random wandering behavior with a more limited range, such as the unemployment rate and
interest rates, although the historical patterns vary considerably over the sample period. Our approach to the analysis of the time series in these databases is agnostic, using the HP, bHP and HRF methodologies to filter each
individual time series without any pre-processing.
Regardless of the sample size, we maintain the use of the standard setting $\lambda=1600$, just as in the simulations.

These time series have their own distinctive scales of measurement. GDP, for instance, is in
trillions of US dollars whereas interest rates are reported in small percentage terms.
To make the time series comparable for group assessment we take the filtered residuals
and standardize each estimated residual sequence to have unit sample
variance. These standardized residuals provide rough individual measures of the nature and form of business
cycles over the historical period.  

\begin{figure}
\includegraphics[scale=0.4]{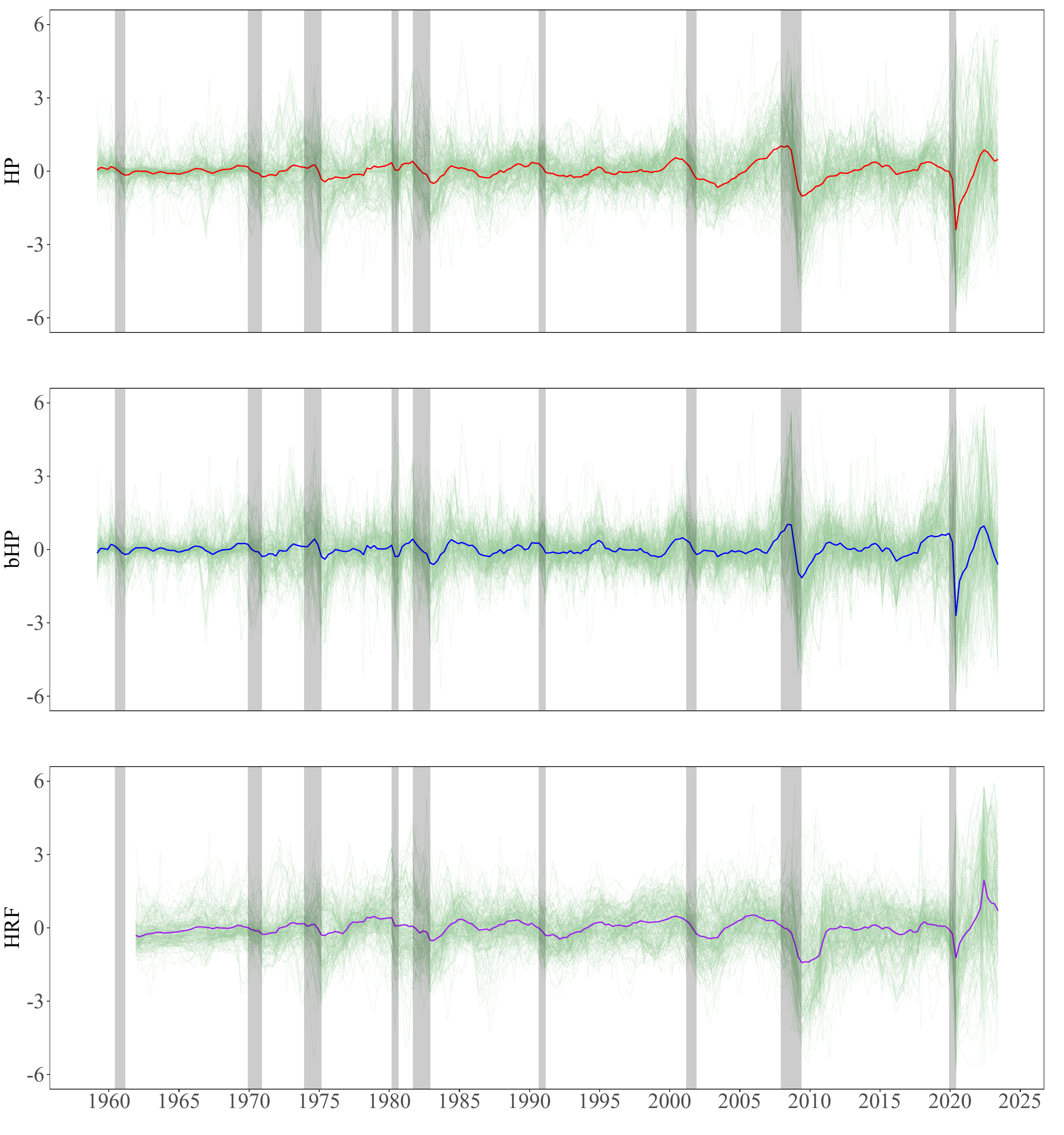}

\caption{\label{fig:Aggregate-q} Aggregate Cyclical Indices from FRED-QD}
\end{figure}

Of particular interest are the shapes of the estimated cycles around periods of 
recession. Some of the time series require reorientation for this purpose as their natural movements reverse normal directions during recessions, e.g., by uptrending in recession periods. These are manually identified and their signs changed. Thus, while most variables decline as economic activity diminishes, unemployment rates rise and so flipping the sign of the unemployment rate to negative brings the direction of movement in line with the majority of the other variables. Accordingly, upon inspection minus signs are assigned to ID 58-72 and 197 
(unemployment) and ID 158-162 (money stocks) in the database.


The mass of thin green lines in Figure \ref{fig:Aggregate-q} show the cyclical
components of the time series as estimated by the HP filter (upper panel), the 
bHP (middle panel), and the HRF (lower panel).
These individual residual series are evidently noisy, and the collection of 246 lines in a single graph merge together into a dense green shading where individual movements are hardly visible.
We add solid
lines (red for HP, blue for bHP, violet for HRF) to aggregate the individual sequences by simple cross section 
averaging into ensemble indices. These indices provide summary indicator measures of the business cycle status of
the economy that broadly reflect the movements of the individual time series estimates. 

The average indices generated by the HP and bHP filters
move in a fairly synchronized way, particularly around the recession periods identified by the NBER which are marked by the grey shaded areas.
For instance, after the dot-com bubble burst in March 2000 these indices each decline to the trough during the 2001 recession of the early 2000s, matching the contraction phase of the business cycle. The bHP index then begins to rise again with the recovery, whereas the HP index is slower and takes much longer to recover. For the great recession of 2008--2009, these two indices again decline to a trough as the financial sector turmoil spreads through the real economy during the economic contraction and then subsequently go through a slow return to normalcy matching the slow recovery with the HP index again taking longer to recover. The recent Covid-19 recession was of much greater scale and the recovery followed swiftly. Outside of
the recession periods, business and growth cycles are still visible in the movement of the indices, such as 
the long expansion in 1990s and the strong economic activity prior to the Covid pandemic.
Overall, from the historical data these two indices appear in line with consensus perceptions of US economic activity. Although the differences between these two indices are subtle, the bHP
filtered index seems more responsive to historical movements in the data and more in line with the chronology of the recession phases than the HP filter. 
The results for the HRF again feature the
enormous peak in 2022:Q2, and the lagged recovery after the dot-com bubble and GFC. 


\begin{figure}
\includegraphics[scale=0.4]{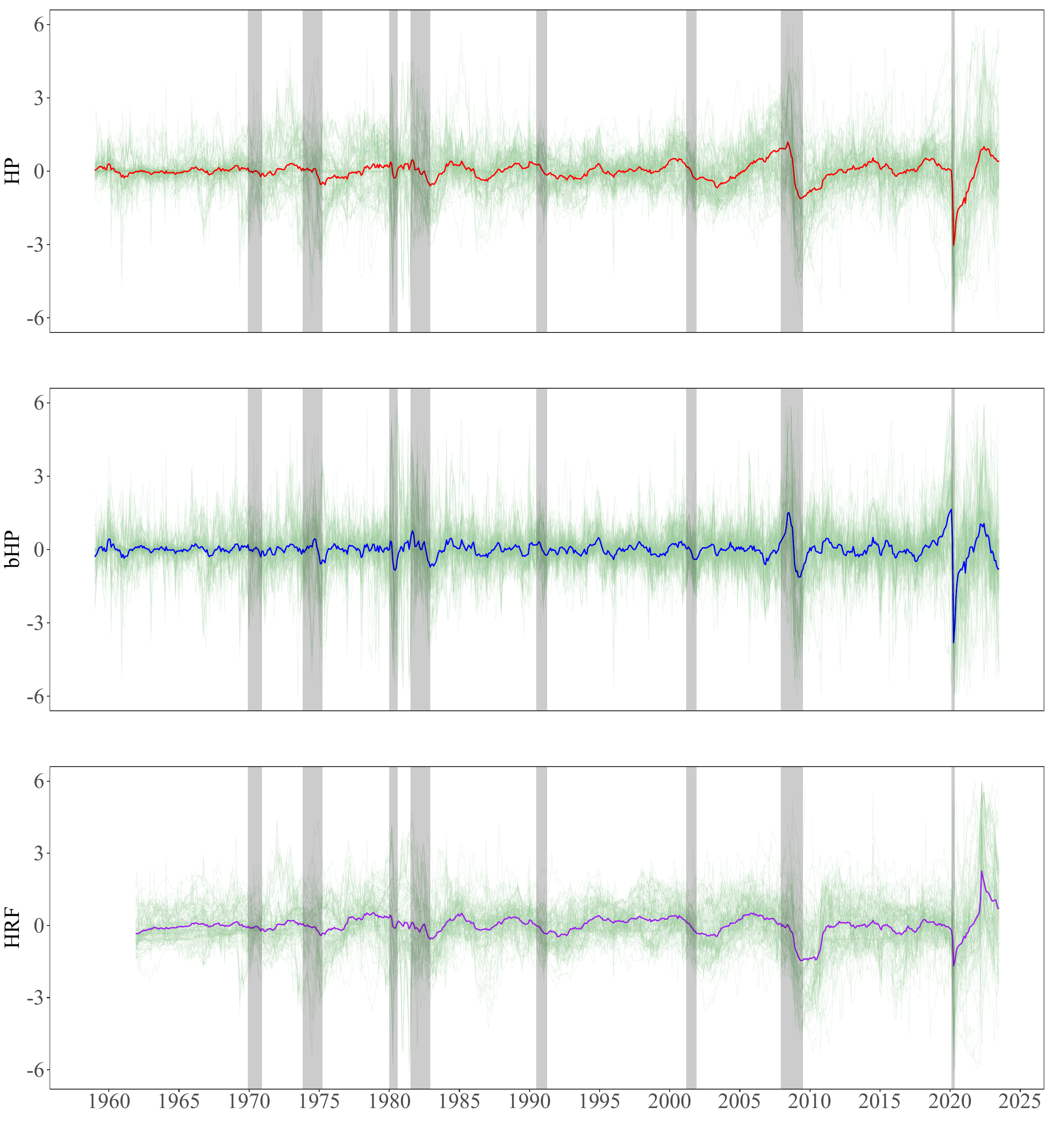}
\caption{\label{fig:Aggregate-m} Aggregate Cyclical Indices from FRED-MD}{}
\end{figure}

To check the robustness of these findings from quarterly data, parallel exercises are
conducted for the 127 individual time series in the monthly FRED-MD database. The tuning parameter $\lambda=129600$ is used for HP and bHP, and a regression filter with $(h,p)=(24,12)$ is employed. Negative signs are given to series 25--31 (unemployment) and
series 70--73 (money stock). The scale-standardized residuals
are plotted in Figure \ref{fig:Aggregate-m} again in thin green lines and the colored solid lines are simple averages of these. 
Similar patterns are observed from these monthly data and the filtered indices
are now evident with fine grain monthly movements. For example, 
in the recovery episode after Covid-19,
the aggregate HP filtered index are relatively smooth whereas the bHP index shows
small fluctuations that accord with waves of new Covid variants and the effects of lockdowns on economic activity. 



\section{Conclusion} 

This paper extends the analysis of the boosted HP filter
to higher order integrated processes and time series with roots that are local to unity, possibly coupled with polynomial time trends and structural breaks. 
The primary asymptotic effect of the HP filter is to smooth the underlying stochastic trend, a property that typically leads to inconsistent estimation but still provides a general picture of how the trend has evolved historically. Boosting the HP filter by repeated application ensures consistent estimation of the full limiting trend process and allows for a rich combination of possible trend behavior including breaks. Practical implementation is facilitated by a data-driven procedure that gives researchers an automated machine learning tool of empirical analysis. 

The HP and bHP methods of trend estimation are motivated by the penalized likelihood approach initiated by \cite{whittaker1922new} which seeks in the context of a model such as \eqref{eq:y_decomp} to find the `most probable' trend function $f_t$ over a historical period. Trends, as is now well understood, encompass a vast range of possible behavior in which future behavior is not always predictable from the past. It is for this reason that economists are so often wrong in assessing macroeconomic and financial market tendencies, to wit: whether inflation will be transient or sustained, whether there is stock market exuberance and potentially serious consequences of collapse,\footnote{Recall Queen Elizabeth II's timely and famously insightful comment at the London School of Economics (5 November, 2008) on the 2008 GFC collapse that `It's awful --- why did nobody see it coming?'
}
when a recession will occur, how extensive it will be, how long a recovery will take and whether an earlier trend path will be resumed. In a model of the form \eqref{eq:y_decomp} representing behavior over an historical period, such characteristics are implicitly embodied.


The alternative paradigm of pure predictive modeling was strongly advocated by \cite{hamilton2018you} and motivated by the following characterization:

\begin{quotation}
	``Here I suggest an alternative concept of what we might mean by
	the cyclical component of a possibly nonstationary series: How different
	is the value at date $t+h$ from the value that we would have expected
	to see based on its behavior through date $t$?'' \citep[p.~836]{hamilton2018you}
\end{quotation}

\begin{figure}
\includegraphics[scale=0.4]{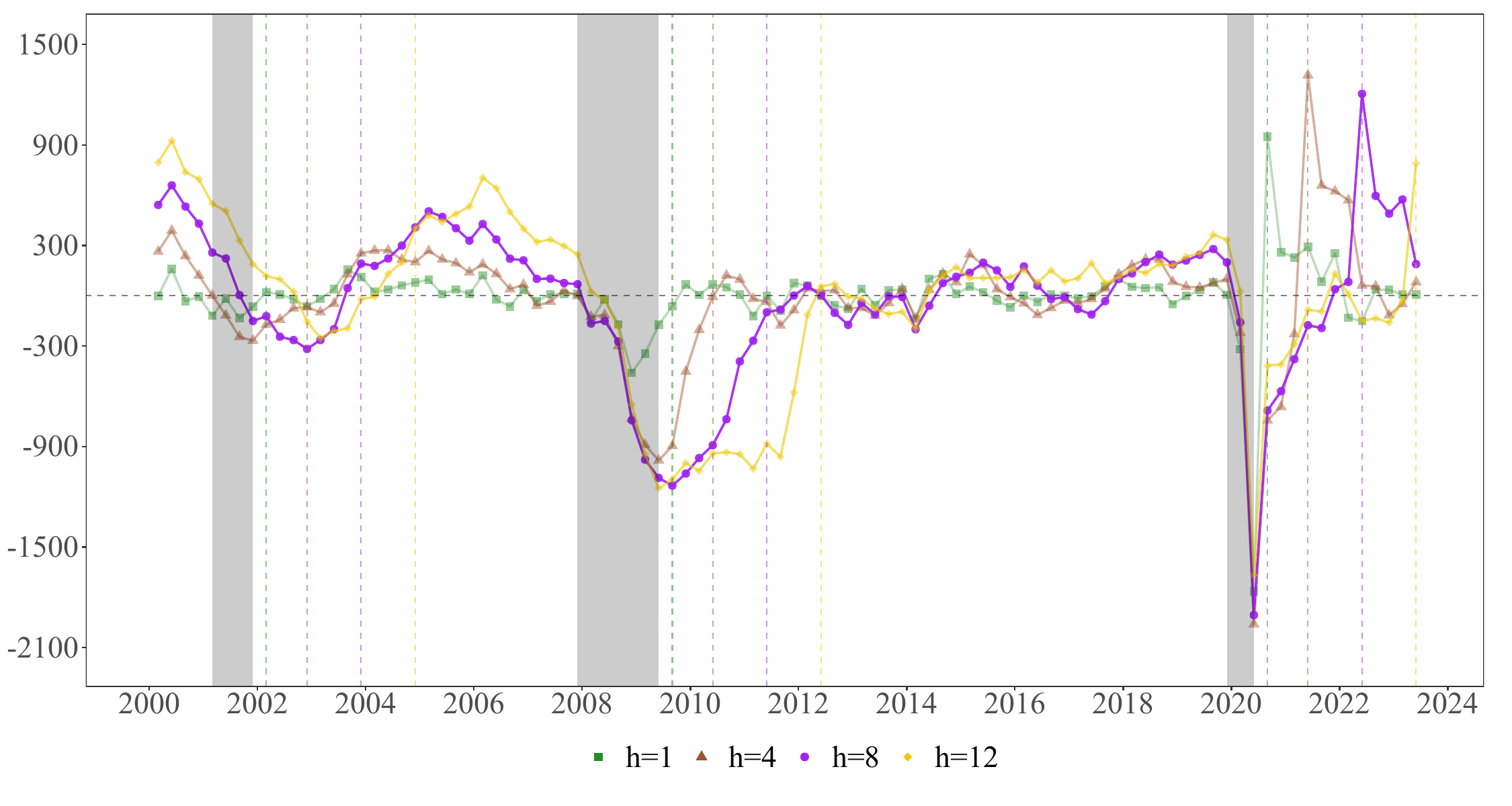}
\caption{\label{fig:HRF-h} Cyclical Component of US GDP under $h=1, 4, 8$, and $12$}{}
\end{figure}

\noindent This conceptualization acknowledges the possibility of model misspecification in that future paths, including cycles but also trends, may differ from those that may be anticipated from past observations. In doing so, the view implicitly accepts that predictive models such as autoregressions may well be incompatible with historical models such as (\ref{eq:y_decomp}) that seek to graduate the data to understand the trends and cycles that have occurred in the past. 
The choice of $h$, based on the researcher's prior knowledge of length of the business cycles, plays a fundamental role. In Figure \ref{fig:HRF-h}, we plot the estimated cyclical component of the US GDP under $h=1, 4, 8,$ and $12$. When $h=8$, the curve is exactly the same as the bottom-right panel of Figure \ref{fig:GDP-q-21}.
The dotted vertical lines are $h$-period after the end of the recessions, and the colors are associated with those of the curves. 
According to this graph, how long the economy recovers from the recession is not determined by the data, but by the researcher's choice of $h$. This feature is particularly pronounced in the enormous peaks after the Covid-19 recession.

Trends and cycles in real world economies are complex and often evolve in unanticipated ways even though the mechanisms and behaviors that drive them may well have some common characteristics \citep{reinhart2009time}. As \citetalias{phillips2021boosting} (p.~551) remarked:

\begin{quotation}
 ``... the period and intensity of business cycles and recessions are so noted for their irregularity that these features are embodied in the many popular descriptive terminologies that are given to them, among which we may mention the terms great depression, great moderation, great recession, short sharp recession, and long recovery.'' 
\end{quotation} 

Such phenomena are difficult to faithfully capture by a supervised learning method like an autoregression that uses only lagged observations to predict future outcomes. 
By contrast, bHP is a nonparametric unsupervised machine learning approach
that is capable of extracting a wide class of underlying trends and cycles from the data. In doing so, the method accommodates historical decomposition of the type \eqref{eq:y_decomp} and recognizes the existence of underlying trend and cycle formations that may take irregular and unpredictable forms. 
The asymptotic theory, simulations and empirical applications reported here all corroborate these advantages.

\bigskip

{\singlespacing
\footnotesize
\bibliographystyle{chicagoa}
\bibliography{bhp2}
}
\processdelayedfloats
\newpage{}

\setcounter{footnote}{0}
\setcounter{table}{0} 
\setcounter{figure}{0} 
\setcounter{equation}{0} 

\renewcommand{\thefootnote}{\thesection.\arabic{footnote}} 
\renewcommand{\theequation}{\thesection.\arabic{equation}} 
\renewcommand{\thefigure}{\thesection.\arabic{figure}} 
\renewcommand{\thetable}{\thesection.\arabic{table}} 

\renewcommand*{\thefootnote}{\fnsymbol{footnote}}

\newtheorem{theorem}{Theorem}[section]
\newtheorem{lemma}{Lemma}[section]
\newtheorem{corollary}{Corollary}[section]

\theoremstyle{remark}
\newtheorem{remark}{Remark}[section]

\newpage
\onehalfspacing
\appendix

\bigskip 
\bigskip 

\begin{center}
\LARGE Online Appendix for ``The boosted HP filter is \\more general than you might think''\footnote[1]{ Phillips acknowledges research support from a Kelly Fellowship at the University of Auckland. Shi thanks the National Natural Science Foundation of China (NSFC) for financial support under the grant number 72133002.}
\end{center}

\bigskip

\begin{center}
\large{
Ziwei Mei$^{a}$, Peter C. B. Phillips$^{b,c,d,e}$ and Zhentao Shi$^{a}$ \vspace{2mm}
\\ \small $^{a}$The Chinese University of Hong Kong 
\\ \small $^{b}$University of Auckland, $^{c}$Yale University, 
\\ \small $^{d}$Singapore Management University, $^{e}$University of Southampton 
}
\end{center}

\bigskip

\begin{quotation} 
This supplementary document has two sections. Section \ref{sec:proof} provides proofs of the theoretical statements in the main text. Section \ref{sec:numerical} corroborates by numerical simulations the asymptotic implications of Proposition \ref{pop:hpltu}, and reports the MSEs under the $t(5)$ errors.
\end{quotation} 

\section{Proofs}\label{sec:proof}

\subsection{Preparatory Results}

To simplify the proof of Lemma \ref{lem:m-operator}, we start with the following preliminary result for the simple case where $a\neq0$ and $m=1$. For any $a\in\mathcal{A}\left(n\right)\backslash\{0\}$,
define the two complex numbers 
\begin{align*}
\delta & ={\rm e}^{2a/n}\left[\frac{n}{a}\left(1-{\rm e}^{-a/n}\right)\right]^{4}-1,\\
\zeta & =\dfrac{\mu a^{4}\left(1+\delta\right)}{\mu a^{4}\left(1+\delta\right)+1}-\dfrac{\mu a{}^{4}}{1+\mu a{}^{4}}.
\end{align*}

\begin{lemma}\label{lem:lem2} Let $\lambda=\mu n^{4}$ and $a\in\mathcal{A}\left(n\right)\backslash\{0\}$.
Then 
\begin{equation}
\left(1-G_{\lambda}\right){\rm e}^{at/n}=\left(\dfrac{\mu a^{4}}{\mu a^{4}+1}+\zeta\right){\rm e}^{at/n},\label{eq:Lem2-1}
\end{equation}
 and the real part and the imaginary part of $\zeta$ are both of
order $O\left(|a|/n\right)$. 
\end{lemma}

\begin{proof}[Proof of Lemma \ref{lem:lem2}]
 Using the operator calculus in \citetalias{phillips2021business} (p.~510) we have 
\begin{eqnarray}
G_{\lambda}{\rm e}^{at/n} & = & \dfrac{1}{1+\mu\mathbb{L}^{-2}[n(1-\mathbb{L})]^{4}}{\rm e}^{at/n}=\int_{0}^{\infty}{\rm e}^{-\{\mu\mathbb{L}^{-2}[n(1-\mathbb{L})]^{4}+1\}s}{\rm e}^{at/n}{\rm d}s\nonumber \\
 & = & \int_{0}^{\infty}{\rm e^{-s}}\sum_{j=0}^{\infty}\dfrac{(-s)^{j}}{j!}[\mu\mathbb{L}^{-2}[n(1-\mathbb{L})]^{4}]^{j}{\rm e}^{at/n}{\rm d}s.\label{eq:G_lam}
\end{eqnarray}
Notice that 
\begin{eqnarray*}
 &  & [\mu\mathbb{L}^{-2}[n(1-\mathbb{L})]^{4}]^{j}{\rm e}^{at/n}\\
 & = & \mu^{j}n^{4j}(1-\mathbb{L})^{4j}{\mathrm{e}}^{a(t+2j)/n}
 =\mu^{j}n^{4j}\sum_{k=0}^{4j}(-1)^{k}{4j \choose k}{\mathrm{e}}^{a(t+2j-k)/n}\\
 & = & \mu^{j}n^{4j}{\rm e}^{a(t+2j)/n}\sum_{k=0}^{4j}(-1)^{k}{4j \choose k}{\rm e}^{-ak/n}=\mu^{j}n^{4j}{\rm e}^{a(t+2j)/n}(1-{\rm e}^{-a/n})^{4j}\\
 & = & \left[\mu{\rm e}^{2a/n}n^{4}(1-{\rm e}^{-a/n})^{4}\right]^{j}{\rm e}^{at/n},
\end{eqnarray*}
using the binomial expansion $(x+y)^{4j}=\sum_{k=0}^{4j}{4j \choose k}x^{k}y{}^{4j-k}$. Evaluation of (\ref{eq:G_lam}) gives 
\begin{align*}
G_{\lambda}{\rm e}^{at/n} & =\int_{0}^{\infty}{\rm e^{-s}}\sum_{j=0}^{\infty}\dfrac{(-s)^{j}}{j!}\left[\mu{\rm e}^{2a/n}n^{4}(1-{\rm e}^{-a/n})^{4}\right]^{j}{\rm e}^{at/n}{\rm d}s \\
 & =\int_{0}^{\infty}{\rm e}^{-\{\mu{\rm e}^{2a/n}n^{4}(1-{\rm e}^{-a/n})^{4}+1\}s}{\rm e}^{at/n}{\rm d}s \\
 & =\text{\ensuremath{\dfrac{{\rm e}^{at/n}}{\mu{\rm e}^{2a/n}n^{4}(1-{\rm e}^{-a/n})^{4}+1}}} 
   =\text{\ensuremath{\dfrac{{\rm e}^{at/n}}{\mu a^{4}\left(1+\delta\right)+1}}},
\end{align*}
by the definition of $\delta$. The residual operator yields
\[
\left(1-G_{\lambda}\right){\rm e}^{at/n}=\left(\dfrac{\mu a^{4}\left(1+\delta\right)}{\mu a^{4}\left(1+\delta\right)+1}\right){\rm e}^{at/n}=\left(\dfrac{\mu a^{4}}{\mu a^{4}+1}+\zeta\right){\rm e}^{at/n}
\]
by the definition of $\zeta$, verifying (\ref{eq:Lem2-1}).

To establish the order of $\zeta$, we first check the order of
$\delta$ by working on its two factors ${\rm e}^{a/n}$ and $\frac{n}{a}\left(1-{\rm e}^{-a/n}\right)$
one by one. Taylor expansion of ${\rm e}^{a/n}$ for $a \in \mathcal{A}(n)$
gives 
\[
{\rm e}^{a/n}=1+\sum_{j=1}^{\infty}\dfrac{1}{j!}\left(\dfrac{a}{n}\right)^{j}=1+O\left(\frac{\left|a\right|}{n}\right).
\]
Similar expansion of ${\rm e}^{-a/n}$ gives 
\begin{align*}
\frac{n}{a}\left(1-{\rm e}^{-a/n}\right) & =\frac{n}{a}\left(1-\sum_{j=0}^{\infty}\dfrac{(-1)^{j}}{j!}\left(\dfrac{a}{n}\right)^{j}\right)=-\sum_{j=1}^{\infty}\dfrac{(-1)^{j}}{j!}\left(\dfrac{a}{n}\right)^{j-1}\\
 & =1-\sum_{j=2}^{\infty}\dfrac{(-1)^{j}}{j!}\left(\dfrac{a}{n}\right)^{j-1}=1+O\left(\frac{\left|a\right|}{n}\right).
\end{align*}
Substituting these terms into $\delta$ we have 
\begin{align}
\delta & =\left[{\rm e}^{a/n}\right]^{2}\left[\frac{n}{a}\left(1-{\rm e}^{-a/n}\right)\right]^{4}-1=\left(1+O\left(\dfrac{|a|}{n}\right)\right)^{6}-1
 =O\left(\dfrac{|a|}{n}\right).\label{eq:diffRate}
\end{align}
Next, decomposing $\zeta$ we have
\begin{align*}
\zeta & =\dfrac{\mu a{}^{4}\delta}{\mu a^{4}\left(1+\delta\right)+1}+\left(\dfrac{\mu a{}^{4}}{\mu a^{4}\left(1+\delta\right)+1}-\dfrac{\mu a{}^{4}}{1+\mu a{}^{4}}\right)\\
 & =\dfrac{\delta}{1+\delta+1/(\mu a{}^{4})}+\left(\dfrac{1}{1+\delta+1/(\mu a{}^{4})}-\dfrac{1}{1+1/(\mu a{}^{4})}\right)=:s_{1}+s_{2}.
\end{align*}
By the triangle inequality and \eqref{eq:diffRate}, the modulus of the denominator of $s_{1}$ satisfies 
\[
|1+\delta+1/(\mu a{}^{4})|\geq\left|1+1/(\mu a^{4})\right|-\left|\delta\right|\geq1-\left|\delta\right|=1-O\left(\left|a\right|/n\right)\geq1/2,
\]
for $n$ sufficiently large, 
bounded away from the origin. 
So the modulus of $s_{1}$ has order no larger than its numerator, which implies that $|s_{1}|=O(\delta)=O(\left|a\right|/n)$.
Similarly, we have 
\[
\left|s_{2}\right|=\left|\dfrac{\delta}{\left(1+\delta+1/(\mu a{}^{4})\right)\left(1+1/(\mu a{}^{4})\right)}\right|\leq\left|\dfrac{\delta}{1+\delta+1/(\mu a{}^{4})}\right|=\left|s_{1}\right|=O\left(\left|a\right|/n\right),
\]
thereby confirming that $\zeta=O\left(|a|/n\right)$
since $\left|\zeta\right|\leq\left|s_{1}\right|+\left|s_{2}\right|=O\left(\left|a\right|/n\right)$. 
\end{proof}
\medskip

We extends the operator algebra in
Lemma \ref{lem:lem2} to general $m$ to  prove Lemma \ref{lem:m-operator}. 

\begin{proof}[Proof of Lemma \ref{lem:m-operator}]
When $a=0,$ we have for $m \ge 1$
\[
\left(1-G_{\lambda}\right)^{m}\mathrm{e}^{at/n}=\left(1-G_{\lambda}\right)^{m}1=\left(\dfrac{\mu n^{4}\mathbb{L}^{-2}}{1+\mu\mathbb{L}^{-2}[n(1-\mathbb{L})]^{4}}\right)^{m}(1-\mathbb{L}){}^{4m}1=0,
\]
as well as 
\begin{align*}
\left[\left(1-G_{\lambda}\right)^{m}-\left(\frac{\mu a^{4}}{\mu a^{4}+1}\right)^{m}\right]\mathrm{e}^{at/n} & =\left[\left(1-G_{\lambda}\right)^{m}-0\right]1=0.
\end{align*}
We now focus on the case when $a\neq0.$ 

\textbf{Part (a}). When the operator $(1-G_{\lambda})$ is repeatedly
applied $m$ times with $m$ a fixed integer, (\ref{eq:Lem2-1})
yields
\begin{align}
(1-G_{\lambda})^{m}{\rm e}^{at/n} & =\left[\dfrac{\mu a{}^{4}}{1+\mu a{}^{4}}+\zeta\right]^{m}{\rm e}^{at/n}\nonumber \\
 & =\left[\left(\dfrac{\mu a{}^{4}}{1+\mu a{}^{4}}\right)^{m}+\sum_{j=1}^{m}{m \choose j}\left(\dfrac{\mu a{}^{4}}{1+\mu a{}^{4}}\right)^{m-j}\zeta^{j}\right]{\rm e}^{at/n}, \label{eq:lem1-m}
\end{align}
where the second equality holds by binomial expansion. Rearranging
this equation and taking the modulus gives
\begin{eqnarray*}
 &  & \left|\left[(1-G_{\lambda})^{m}-\left(\dfrac{\mu a{}^{4}}{1+\mu a{}^{4}}\right)^{m}\right]{\rm e}^{at/n}\right|\\
 & = & \left|\sum_{j=1}^{m}{m \choose j}\left(\dfrac{\mu a{}^{4}}{1+\mu a{}^{4}}\right)^{m-j}\zeta^{j}{\rm e}^{at/n}\right|\leq\left|\sum_{j=1}^{m}{m \choose j}\left(\dfrac{\mu a{}^{4}}{1+\mu a{}^{4}}\right)^{m-j}\zeta^{j}\right|\left|{\rm e}^{at/n}\right|\\
 & \leq & {\rm e}^{\left|a\right|}\sum_{j=1}^{m}{m \choose j}\left(\dfrac{\mu a{}^{4}}{1+\mu a{}^{4}}\right)^{m-j}\left|\zeta\right|^{j}\leq{\rm e}^{\left|a\right|}\sum_{j=1}^{m}{m \choose j}\left|\zeta\right|^{j}\\
 & \leq & m\left|\zeta\right|{\rm e}^{\left|a\right|}\left(1+o\left(1\right)\right)=O\left(n^{-1}|a|{\rm e}^{\left|a\right|}\right),
\end{eqnarray*}
as $t \le n$, $0<\dfrac{\mu a{}^{4}}{1+\mu a{}^{4}}\leq1$
and $\zeta=O\left(|a|/n\right)$. Since $|a|\leq\sqrt{\log n}$ for any $a\in\mathcal{A}\left(n\right)$,
it follows that  
\begin{align*}
n^{-1}|a|{\rm e}^{\left|a\right|} & \leq n^{-1}\sqrt{\log n}\exp\left(\sqrt{\log n}\right)\\
 & =\exp\left(\sqrt{\log n}+\log\sqrt{\log n}-\log n\right)\to  0
\end{align*}
as $n\to\infty$, giving Part (a) for fixed $m$. 

\textbf{Part (b}). Taking the modulus of (\ref{eq:lem1-m}), we have
\begin{eqnarray*}
\left|(1-G_{\lambda})^{m}{\rm e}^{at/n}\right| & = & \left|\left[\dfrac{\mu a{}^{4}}{1+\mu a{}^{4}}+\zeta\right]^{m}{\rm e}^{at/n}\right|\leq\left|\dfrac{\mu a{}^{4}}{1+\mu a{}^{4}}+\zeta\right|^{m}\left|{\rm e}^{at/n}\right|\\
 & \leq & \left|\dfrac{\mu a{}^{4}}{1+\mu a{}^{4}}+\left|\zeta\right|\right|^{m}{\rm e}^{\left|a\right|}=\left|1-\dfrac{1-|\zeta|(1+\mu a{}^{4})}{1+\mu a{}^{4}}\right|^{m}{\rm e}^{\left|a\right|}.
\end{eqnarray*}
For any $a\in\mathcal{A}\left(n\wedge m\right)\subseteq\mathcal{A}\left(n\right)$
the order of $\zeta$ in Lemma \ref{lem:lem2} remains valid and thus
\[
1-|\zeta|(1+\mu a{}^{4})=1-O(|a|/n+|a|^{5}/n)\geq1/2
\]
for $m,n$ sufficiently large. The above inequality, 
together with $\left|a\right|\leq\log\left(m\wedge n\right)\leq\log m$,
implies 
\begin{align*}
\left|(1-G_{\lambda})^{m}{\rm e}^{at/n}\right| & \leq\left(1-\dfrac{1/2}{1+\mu\left(\log m\right)^{2}}\right)^{m}{\rm e}^{\sqrt{\log m}}
\end{align*}
uniformly for $a\in\mathcal{A}(n\wedge m)$ and $t\leq n$ with $m,n$ sufficiently large, so that 
\[
\left(1-\dfrac{1/2}{1+\mu\left(\log m\right)^{2}}\right)^{m}{\rm e}^{\sqrt{\log m}}\to\lim_{m\to\infty}\exp\left(-\dfrac{m/2}{1+\mu\left(\log m\right)^{2}}+\sqrt{\log m}\right)=0,
\]
giving the stated result for Part (b).
\end{proof}
\medskip

The following Corollary \ref{cor R_sin_cos} is an immediate implication
of Lemma \ref{lem:m-operator} as the HP residual operator is repeatedly
applied to trigonometric functions with increasingly higher frequencies.

\begin{corollary} \label{cor R_sin_cos} 

Suppose $\lambda=\mu n^{4}$.
\begin{enumerate}
\item For any fixed $m\in\mathbb{N},$ if $K_{n}=\left\lfloor \pi^{-1}\sqrt{\log n}\right\rfloor $,
then
\begin{eqnarray*}
\sup_{1\leq t\leq n,\,k\leq K_{n}}\left|\left[\left(1-G_{\lambda}\right)^{m}-\left(\dfrac{\mu}{\mu+\lambda_{k}^{2}}\right)^{m}\right]\varphi_{k}\left(\dfrac{t}{n}\right)\right| & \to & 0 
\\
\sup_{1\leq t\leq n,\,k\leq K_{n}}\left|\left[\left(1-G_{\lambda}\right)^{m}-\left(\dfrac{\mu}{\mu+\lambda_{k}^{2}}\right)^{m}\right]\psi_{k}\left(\dfrac{t}{n}\right)\right| & \to & 0 
\end{eqnarray*}
 as $n\to\infty$. 
\item If $K_{n,m}=\left\lfloor \pi^{-1}\sqrt{\log(n\wedge m)}\right\rfloor $,
then 
\begin{eqnarray*}
\sup_{1\leq t\leq n,\,k\leq K_{n,m}}\left|\left(1-G_{\lambda}\right)^{m}\varphi_{k}(t/n) \right| & \to & 0\\\ 
\sup_{1\leq t\leq n,\,k\leq K_{n,m}}\left|\left(1-G_{\lambda}\right)^{m}\psi_{k}(t/n)\right| & \to & 0
\end{eqnarray*}
as $n,m\to\infty$.
\end{enumerate}
\end{corollary}
\begin{proof}[Proof of Corollary \ref{cor R_sin_cos}]
 \textbf{Part (a}). The definitions of $\psi_{k}\left(\cdot\right)$
and $\varphi_{k}\left(\cdot\right)$ give 
\begin{equation*}
\left(1-G_{\lambda}\right)^{m}\left[\psi_{k}\left(\dfrac{t}{n}\right)+\mathbf{i}\varphi_{k}\left(\dfrac{t}{n}\right)\right]=\sqrt{2}\left(1-G_{\lambda}\right)^{m}\mathrm{e}^{\frac{\mathbf{i}(t/n)}{\sqrt{\lambda_{k}}}}. 
\end{equation*}
Let $a=\mathbf{i}/\sqrt{\lambda_{k}}$. We verify 
\[
a^{4}=\lambda_{k}^{-2}=\left[\left(k-1/2\right)\pi\right]^{4}\leq K_{n}^{4}\pi^{4}\leq\left(\log n\right)^{2}
\]
satisfies the condition $a\in\mathcal{A}(n)$, and then Lemma \ref{lem:m-operator}
(a) ensures that for any fixed $m$
\begin{equation*}
\sup_{1\leq t\leq n,\,k\leq K_{n}}\left|\left[\left(1-G_{\lambda}\right)^{m}-\left(\dfrac{\mu}{\mu+\lambda_{k}^{2}}\right)^{m}\right]\mathrm{e}^{\frac{\mathbf{i}(t/n)}{\sqrt{\lambda_{k}}}}\right|\to0 
\end{equation*}
 as $n\to\infty.$ 
We complete the proof 
by separating the imaginary and the real parts of $\exp\left(\frac{\mathbf{i}(t/n)}{\sqrt{\lambda_{k}}}\right)$,
respectively. 

\textbf{Part (b}). Similarly, for $a=\mathbf{i}/\sqrt{\lambda_{k}}$
we verify that $a^{4}\leq K_{n,m}^{4}\pi^{4}\leq\left(\log\left(n\wedge m\right)\right)^{2}$.
The fact $a\in\mathcal{A}(n\wedge m)$ allows us to invoke Lemma \ref{lem:m-operator}
(b):
\begin{equation}
\sup_{1\leq t\leq n\,k\leq K_{n,m}}\left|\left(1-G_{\lambda}\right)^{m}\mathrm{e}^{\frac{\mathbf{i}(t/n)}{\sqrt{\lambda_{k}}}}\right|\to0  \label{eq:mconv_exp-1}
\end{equation}
as $m,n\to\infty$, and then the results 
follow.
\end{proof}

\begin{remark}
\label{rem:three_HP} When setting $m=1$, Lemma \ref{lem:m-operator}
and Corollary \ref{cor R_sin_cos} immediately imply 
\begin{equation*}
\left(G_{\lambda}-\dfrac{1}{\mu c^{4}+1}\right)\mathrm{e}^{ct/n}\to 0 
\end{equation*}
for any $c\in\mathbb{\mathbb{R}}$
uniformly over all $t\leq n$,
and 
\begin{eqnarray}
\left(G_{\lambda}-\dfrac{\lambda_{k}^{2}}{\mu+\lambda_{k}^{2}}\right)\varphi_{k}\left(\dfrac{t}{n}\right) & \to & 0, \\
\nonumber
\left(G_{\lambda}-\dfrac{\lambda_{k}^{2}}{\mu+\lambda_{k}^{2}}\right)\psi_{k}\left(\dfrac{t}{n}\right) & \to & 0.\label{eq:conv_psi-2}
\end{eqnarray}
uniformly for all $k\leq K_{n}$ under consideration. These real exponential
functions, sine waves and cosine waves are the building blocks of
the series representations of the higher order integrated processes
and LUR processes.
\end{remark}

\subsection{Main Results}
\begin{proof}[Proof of Proposition \ref{pop:hpI2}]
 The KL representation of $B_{2}(r)$ specified in (\ref{eq:L2})
converges almost surely and uniformly in $[0,1]$. Let the $K_{n}$-term
finite KL representation be $B_{2}^{K_{n}}(r):=\sum\limits _{k=1}^{K_{n}}\lambda_{k}(\sqrt{2}-\psi_{k}(r))\xi_{k}$.
When $K_{n}\to\infty$ as $n\to\infty$, $\sup\limits _{0\leq t\leq n}\left|B_{2}(t/n)-B_{2}^{K_{n}}(t/n)\right|=o_{a.s.}(1)$
and by uniform convergence 
\begin{equation*}
\sup\limits _{0\leq t\leq n}\left|X_{n}\left(t/n\right)-B_{2}^{K_{n}}\left(t/n\right)\right|=o_{a.s.}(1).
\end{equation*}
It follows that $X_{n}(t/n)$ is almost surely uniformly well approximated
by $B_{2}^{K_{n}}(t/n)$ for $t\leq n$ as $n\to\infty$. Hence the HP
filtered trend has the following approximation 
\begin{align}
\frac{\widehat{f}_{t}^{\mathrm{HP}}}{n^{3/2}} & =G_{\lambda}\frac{x_{t}}{n^{3/2}}=G_{\lambda}\left[B_{2}^{K_{n}}\left(\frac{t}{n}\right)+o_{a.s.}(1)\right]\nonumber \\
 & =\sum_{k=1}^{K_{n}}\lambda_{k}\left[G_{\lambda}\left(\sqrt{2}-\psi_{k}\left(\frac{t}{n}\right)\right)\right]\xi_{k}+o_{a.s.}(1)\nonumber \\
 & =\sum_{k=1}^{K_{n}}\sqrt{2}\lambda_{k}\xi_{k}-\sum_{k=1}^{K_{n}}\lambda_{k}\xi_{k}G_{\lambda}\psi_{k}\Big(\dfrac{t}{n}\Big)+o_{a.s.}(1)\label{eq:fthatI2}
\end{align}
as $n\to\infty$. The $o_{a.s.}(1)$ in (\ref{eq:fthatI2}) holds
because the two sided moving average filter produced by the operator
$G_{\lambda}$ is an absolutely summable weighted moving average with
stable geometric decay, which preserves the error order by majorization. 

The asymptotic form of the HP filter can be written, according to
(\ref{eq:conv_psi-2}), as 
\begin{align*}
\dfrac{\hat{f}_{t,K_{n}}^{\text{HP}}}{n^{3/2}} & =\sum\limits _{k=1}^{K_{n}}\sqrt{2}\lambda_{k}\xi_{k}-\sum\limits _{k=1}^{K_{n}}\lambda_{k}\xi_{k}\left[\dfrac{\lambda_{k}^{2}}{\mu+\lambda_{k}^{2}}\psi_{k}\left(\dfrac{t}{n}\right)+o(1)\right]+o_{a.s.}(1)\\
 & =\sum\limits _{k=1}^{K_{n}}\sqrt{2}\lambda_{k}\xi_{k}-\sum\limits _{k=1}^{K_{n}}\dfrac{\lambda_{k}^{3}}{\mu+\lambda_{k}^{2}}\psi_{k}\left(\dfrac{t}{n}\right)\xi_{k}+o_{a.s.}(1).
\end{align*}
Note that $\lambda_{k}=1/[(k-\frac{1}{2})\pi]^{2}$ and $\frac{\lambda_{k}^{3}}{\mu+\lambda_{k}^{2}}=O(k^{-6})$.
Hence, the series $\sum_{k=1}^{\infty}\lambda_{k}\xi_{k}$ and $\sum_{k=1}^{\infty}\dfrac{\lambda_{k}^{3}}{\mu+\lambda_{k}^{2}}\psi_{k}\left(\frac{t}{n}\right)\xi_{k}$
converge almost surely and uniformly as $K_{n}\to\infty$. When $K_{n}=\left\lfloor \pi^{-1}\sqrt{\log n}\right\rfloor $
as $n\to\infty$, by Corollary \ref{cor R_sin_cos} (a) we obtain the following 
asymptotic form of the HP filter trend as 
\[
\dfrac{\hat{f}_{t}^{\text{HP}}}{n^{3/2}}=\sum\limits _{k=1}^{\infty}\left[\sqrt{2}\lambda_{k}-\dfrac{\lambda_{k}^{3}}{\mu+\lambda_{k}^{2}}\psi_{k}\left(\dfrac{t}{n}\right)\right]\xi_{k}+o_{a.s.}(1).
\]
The proof is completed.
\end{proof}

\begin{proof}[Proof of Theorem \ref{thmbhpIq}]
 When $q=1$, the convergence is already established in Theorem 1
of \citetalias{phillips2021boosting}. In this proof, we focus on $q\geq2.$ By the uniform convergence
law (\ref{uniformIq}) and the KL representation of the $I(q)$ process
in (\ref{eq:Lq}), the bHP estimated cycle has the following approximation
\begin{eqnarray*}
\dfrac{\widehat{c}_{t}^{(m)}}{n^{q-0.5}} & = & \left(1-G_{\lambda}\right)^{m}\dfrac{x_{t}}{n^{q-0.5}}=-\left(1-G_{\lambda}\right)^{m}\left[B_{q}^{K_{n,m}}\left(\dfrac{t}{n}\right)+o_{a.s.}(1)\right] \nonumber \\
 & = & -\sqrt{2}\sum\limits _{k=1}^{K_{n,m}}\xi_{k}(1-G_{\lambda})^{m}\left[\sum_{\ell=1}^{\lfloor q/2\rfloor}(-1)^{\ell-1}\lambda_{k}^{j}\dfrac{(t/n)^{q-2\ell}}{(q-2\ell)!}+\lambda_{k}^{q/2}\text{Im}\left[\left(-\mathbf{i}\right)^{q-1}\mathrm{e}^{\text{\ensuremath{\frac{\mathbf{i}(t/n)}{\sqrt{\lambda_{k}}}}}}\right]\right]\nonumber \\
 &  &  + o_{a.s.}(1) 
\end{eqnarray*}
as $n\to\infty$ with $K_{n,m}=\left\lfloor \pi^{-1}\sqrt{\log(n\wedge m)}\right\rfloor $
by Corollary \ref{cor R_sin_cos} (b). 

When $4m\geq q$ and $1\leq\ell\leq\lfloor q/2\rfloor,$ the polynomial component is 
\begin{eqnarray*}
 &  & \left(1-G_{\lambda}\right)^{m}(t/n)^{q-2\ell}\\
 & = & \dfrac{1}{\Gamma(m)}\int_{0}^{\infty}s^{m-1}\mathrm{e}^{-s\left(1+\lambda\mathbb{L}^{-2}(1-\mathbb{L})^{4}\right)}\left[\lambda\mathbb{L}(1-\mathbb{L})^{4}\right]^{m}\text{\ensuremath{\left(\dfrac{t}{n}\right)}}^{q-2\ell}ds\\
 & = & \dfrac{1}{\Gamma(m)}\int_{0}^{\infty}s^{m-1}\mathrm{e}^{-s}\sum_{j=0}^{\infty}\dfrac{(-1)^{j}s^{j}}{j!}\left[\lambda\mathbb{L}^{-2}(1-\mathbb{L})^{4}\right]^{m+j}\text{\ensuremath{\left(\dfrac{t}{n}\right)}}^{q-2\ell}ds\\
 & = & \dfrac{1}{\Gamma(m)}\int_{0}^{\infty}s^{m-1}\mathrm{e}^{-s}\sum_{j=0}^{\infty}\dfrac{(-1)^{j}s^{j}}{j!}\left[\lambda(1-\mathbb{L})^{4}\right]^{m+j}\text{\ensuremath{\left(\dfrac{t+2(m+j)}{n}\right)}}^{q-2\ell}ds\\
 & = & 0.
\end{eqnarray*}
For the cyclical functions, since  $\text{Im}\big((-\mathbf{i})^{q-1}\mathrm{e}^{\frac{\mathbf{i}(t/n)}{\sqrt{\lambda_{k}}}}\big)$
is either $\pm\cos\left(t/\left(n\sqrt{\lambda_{k}}\right)\right)$
or $\pm\sin\left(t/\left(n\sqrt{\lambda_{k}}\right)\right)$ and $\lambda_{k}^{q/2}=O(k^{-2})$,
the series $\sum\limits _{k=1}^{\infty}\lambda_{k}^{q/2}\xi_{k}$
converges almost surely and uniformly for all $t\leq n$. Thus, according
to (\ref{eq:mconv_exp-1}) we have 
\begin{align*}
 & \sum\limits _{k=1}^{K_{n,m}}\xi_{k}\left(1-G_{\lambda}\right)^{m}\left[\sum_{\ell=1}^{\lfloor q/2\rfloor}(-1)^{\ell-1}\lambda_{k}^{j}\dfrac{(t/n)^{q-2\ell}}{(q-2\ell)!}+\lambda_{k}^{q/2}\text{Im}\left[(-\mathbf{i})^{q-1}\mathrm{e}^{\frac{\mathbf{i}(t/n)}{\sqrt{\lambda_{k}}}}\right]\right]\\
= & \sum\limits _{k=1}^{K_{n,m}}\lambda_{k}^{q/2}\xi_{k}\mathrm{Im}\left((-\mathbf{i})^{q-1}\left[\left(1-G_{\lambda}\right)^{m}\mathrm{e}^{\frac{\mathbf{i}(t/n)}{\sqrt{\lambda_{k}}}}\right]\right)=\sum\limits _{k=1}^{K_{n,m}}\lambda_{k}^{q/2}\xi_{k}\cdot o\left(1\right)=o_{a.s.}\left(1\right)
\end{align*}
when $K_{n,m}=\left\lfloor \pi^{-1}\sqrt{\log(n\wedge m)}\right\rfloor .$
This confirms that $\widehat{c}_{t}^{(m)}/n^{q-0.5}=o_{a.s.}(1)$ uniformly
over $t\leq n$, and thus $n^{0.5-q}\cdot\hat{f}_{\lfloor nr\rfloor}^{(m)}\rightsquigarrow B_{q}(r)$
as stated. 
\end{proof}
\bigskip

Before establishing the results for the LUR case it is convenient to derive the following series representation 
\begin{eqnarray}
J_{c}(r) & = & B(r)+c\int_{0}^{r}\mathrm{e}^{(r-s)c}B(s)ds\nonumber \\
 & = & B(r)+\sqrt{2}c\mathrm{e}^{cr}\sum\limits _{k=1}^{\infty}\xi_{k}\sqrt{\lambda_{k}}\int_{0}^{r}\mathrm{e}^{-sc}\sin\left(\frac{s}{\sqrt{\lambda_{k}}}\right)ds\nonumber \\
 & = & B(r)+\sqrt{2}c\mathrm{e}^{cr}\sum\limits _{k=1}^{\infty}\xi_{k}\frac{\lambda_{k}^{3/2}}{\lambda_{k}c^{2}+1}\left(\frac{1}{\sqrt{\lambda_{k}}}-\mathrm{e}^{-cr}c\sin\left(\frac{r}{\sqrt{\lambda_{k}}}\right)-\frac{\mathrm{e}^{-cr}}{\sqrt{\lambda_{k}}}\cos\left(\frac{r}{\sqrt{\lambda_{k}}}\right)\right)\nonumber \\
 & = & \sqrt{2}\sum\limits _{k=1}^{\infty}\xi_{k}\frac{\sqrt{\lambda_{k}}}{c^{2}\lambda_{k}+1}\sin\left(\frac{r}{\sqrt{\lambda_{k}}}\right)+\sqrt{2}c\sum\limits _{k=1}^{\infty}\xi_{k}\frac{\lambda_{k}}{\lambda_{k}c^{2}+1}\left(\mathrm{e}^{cr}-\cos\left(\frac{r}{\sqrt{\lambda_{k}}}\right)\right)\nonumber \\
 & = & \sum\limits _{k=1}^{\infty}\dfrac{\sqrt{2}c\lambda_{k}\mathrm{e}^{cr}+\sqrt{\lambda_{k}}\varphi_{k}(r)-c\lambda_{k}\psi_{k}(r)}{\lambda_{k}c^{2}+1}\xi_{k}, \label{eq:JCK-detail}
\end{eqnarray}
as in \cite{phillips1998new}. 
This representation is needed in the following proofs.
\begin{proof}[Proof of Proposition \ref{pop:hpltu}]
The series presentation (\ref{eq:JcrKL}) converges almost surely
and uniformly over $r$. It is approximated by the $K_{n}$-term representation 
\[
J_{c}^{K_{n}}(r)=\sum\limits _{k=1}^{K_{n}}\dfrac{\sqrt{2}c\lambda_{k}\mathrm{e}^{cr}+\sqrt{\lambda_{k}}\varphi_{k}(r)-c\lambda_{k}\psi_{k}(r)}{\lambda_{k}c^{2}+1}\xi_{k}
\]
in the sense of $\sup\limits _{0\leq t\leq n}\left|J_{c}(r)-J_{c}^{K_{n}}(r)\right|=o_{a.s.}(1)$
when $K_{n}\to\infty$ as $n\to\infty$, so that in the expanded
space 
$
\sup\limits _{0\leq t\leq n}\left|n^{-1/2}x_{\lfloor nr\rfloor}-J_{c}^{K_{n}}(r)\right|=o_{a.s.}(1)
$
by the uniform convergence (\ref{uniformltu}). The HP estimated trend
is then approximated as 
\begin{equation}
\begin{aligned}\dfrac{\hat{f}_{t}^{\text{HP}}}{n^{1/2}} & =G_{\lambda}\frac{x_{t}}{n^{1/2}}=G_{\lambda}\left[J_{c}^{K_{n}}\left(\frac{t}{n}\right)+o_{a.s.}(1)\right]\\
 & =\sum\limits _{k=1}^{K_{n}}\left[G_{\lambda}\dfrac{\sqrt{2}c\lambda_{k}\mathrm{e}^{ct/n}+\sqrt{\lambda_{k}}\varphi_{k}(t/n)-c\lambda_{k}\psi_{k}(t/n)}{\lambda_{k}c^{2}+1}\right]\xi_{k}+o_{a.s.}(1)\text{.}
\end{aligned}
\label{eq:LURKL}
\end{equation}
In view of Remark \ref{rem:three_HP}, when $K_{n}=\left\lfloor \pi^{-1}\sqrt{\log n}\right\rfloor $
we have
\begin{equation}
\sum\limits _{k=1}^{K_{n}}G_{\lambda}\dfrac{\sqrt{2}c\lambda_{k}\xi_{k}\mathrm{e}^{ct/n}}{\lambda_{k}c^{2}+1}
=\sum\limits _{k=1}^{K_{n}}\dfrac{\sqrt{2}c\lambda_{k}\xi_{k}}{\lambda_{k}c^{2}+1}\left(\dfrac{\mathrm{e}^{ct/n}}{\mu c^{4}+1}+o(1)\right) 
= \frac{\sqrt{2}c\mathrm{e}^{ct/n}}{\mu c^{4}+1}\sum\limits _{k=1}^{K_{n}}\dfrac{\lambda_{k}}{\lambda_{k}c^{2}+1}\xi_{k}+o_{a.s.}(1),\label{eq:LUR1}
\end{equation}
since $\sum\limits _{k=1}^{K_{n}}\dfrac{\lambda_{k}}{\lambda_{k}c^{2}+1}\xi_{k}\sim N\left(0,\omega^{2}\sum\limits _{k=1}^{K_{n}}\dfrac{\lambda_{k}^{2}}{\left(\lambda_{k}c^{2}+1\right)^{2}}\right)$
with variance bounded by 
\begin{equation}
\omega^{2}\sum\limits _{k=1}^{K_{n}}\dfrac{\lambda_{k}^{2}}{\left(\lambda_{k}c^{2}+1\right)^{2}}\leq\omega^{2}\sum\limits _{k=1}^{\infty}\dfrac{\lambda_{k}^{2}}{\left(\lambda_{k}c^{2}+1\right)^{2}}\leq\omega^{2}\sum\limits _{k=1}^{\infty}\lambda_{k}^{2}=\frac{\omega^{2}}{6}.\label{eq:LUR_bound1}
\end{equation}
Similarly, 
\begin{eqnarray}
\sum\limits _{k=1}^{K_{n}}G_{\lambda}\dfrac{\sqrt{\lambda_{k}}\varphi_{k}(t/n)}{\lambda_{k}c^{2}+1} & = & \sum\limits _{k=1}^{K_{n}}\dfrac{\sqrt{\lambda_{k}}\xi_{k}}{\lambda_{k}c^{2}+1}\left(\dfrac{\lambda_{k}^{2}}{\mu+\lambda_{k}^{2}}\varphi_{k}\left(\dfrac{t}{n}\right)+o(1)\right)\nonumber \\
 & = & \sum\limits _{k=1}^{K_{n}}\frac{\lambda_{k}^{2}}{\mu+\lambda_{k}^{2}}\cdot\frac{\sqrt{\lambda_{k}}}{\lambda_{k}c^{2}+1}\varphi_{k}\left(\frac{t}{n}\right)\xi_{k}+o_{a.s.}(1)\label{eq:LUR2}
\end{eqnarray}
and 
\begin{align}
\sum\limits _{k=1}^{K_{n}}G_{\lambda}\dfrac{c\lambda_{k}\xi_{k}\psi_{k}(t/n)}{\lambda_{k}c^{2}+1} & =\sum\limits _{k=1}^{K_{n}}\dfrac{c\lambda_{k}\xi_{k}}{\lambda_{k}c^{2}+1}\left(\dfrac{\lambda_{k}^{2}}{\mu+\lambda_{k}^{2}}\psi_{k}\left(\dfrac{t}{n}\right)+o(1)\right)\nonumber \\
 & =c\sum\limits _{k=1}^{K_{n}}\frac{\lambda_{k}^{2}}{\mu+\lambda_{k}^{2}}\cdot\dfrac{\lambda_{k}}{\lambda_{k}c^{2}+1}\psi_{k}\left(\frac{t}{n}\right)\xi_{k}+o_{a.s.}(1)\label{eq:LUR3}
\end{align}
as $n\to\infty$ by virtue of uniform almost surely convergence. Substituting (\ref{eq:LUR1}),
(\ref{eq:LUR2}), and (\ref{eq:LUR3}) into (\ref{eq:LURKL}) yields
\[
\dfrac{\hat{f}_{t}^{\text{HP}}}{n^{1/2}}=\sum\limits _{k=1}^{K_{n}}\frac{1}{\lambda_{k}c^{2}+1}\left[\dfrac{\sqrt{2}c\lambda_{k}\mathrm{e}^{ct/n}}{\mu c^{4}+1}+\frac{\lambda_{k}^{2}}{\mu+\lambda_{k}^{2}}\left(\sqrt{\lambda_{k}}\varphi_{k}(t/n)-c\lambda_{k}\psi_{k}(t/n)\right)\right]\xi_{k}+o_{a.s.}(1)
\]
uniformly for all $t\leq n$. The limiting expression (\ref{eq:HPLURlimit})
follows as $K_{n}$ passes to infinity as $n\to\infty$. 
\end{proof}

\begin{proof}[Proof of Theorem \ref{thmbhpltu}]
 By virtue of the uniform convergence (\ref{uniformltu}) the estimated
residual in the expanded probability space is 
\begin{align}
\dfrac{\widehat{c}_{t}^{(m)}}{n^{1/2}} & =\left(1-G_{\lambda}\right)^{m}\left[J_{c}^{K_{n,m}}\left(\dfrac{t}{n}\right)+o_{a.s.}(1)\right]\nonumber \\
 & =\sum\limits _{k=1}^{K_{n,m}}\left(1-G_{\lambda}\right)^{m}\dfrac{\sqrt{2}c\lambda_{k}\mathrm{e}^{\frac{ct}{n}}+\sqrt{\lambda_{k}}\varphi_{k}(\frac{t}{n})-c\lambda_{k}\psi_{k}(\frac{t}{n})}{\lambda_{k}c^{2}+1}\xi_{k}+o_{a.s.}(1).\label{eq:mLURKL}
\end{align}
In view of Lemma \ref{lem:m-operator} and Corollary \ref{cor R_sin_cos},
when $K_{n,m}=\left\lfloor \pi^{-1}\sqrt{\log(n\wedge m)}\right\rfloor $
as $m,n\to\infty$, we have 
\begin{align}
\sum\limits _{k=1}^{K_{n,m}}\left(1-G_{\lambda}\right)^{m}\dfrac{\sqrt{2}c\lambda_{k}\xi_{k}\mathrm{e}^{ct/n}}{\lambda_{k}c^{2}+1} & =\sqrt{2}c\sum\limits _{k=1}^{K_{n,m}}\dfrac{\lambda_{k}}{\lambda_{k}c^{2}+1}\xi_{k}\cdot o(1)=o_{a.s.}(1),\label{eq:LURm1}\\
\sum\limits _{k=1}^{K_{n,m}}\left(1-G_{\lambda}\right)^{m}\dfrac{\sqrt{\lambda_{k}}\varphi_{k}(t/n)}{\lambda_{k}c^{2}+1} & =\sum\limits _{k=1}^{K_{n,m}}\dfrac{\sqrt{\lambda_{k}}}{\lambda_{k}c^{2}+1}\xi_{k}\cdot o(1)=o_{a.s.}(1)\text{,}\label{eq:LURm2}\\
\sum\limits _{k=1}^{K_{n,m}}\left(1-G_{\lambda}\right)^{m}\dfrac{c\lambda_{k}\xi_{k}}{\lambda_{k}c^{2}+1}\psi_{k}(\frac{t}{n}) & =c\sum\limits _{k=1}^{K_{n,m}}\dfrac{\lambda_{k}}{\lambda_{k}c^{2}+1}\xi_{k}\cdot o(1)=o_{a.s.}(1), \label{eq:LURm3}
\end{align}
uniformly over $t\leq n$, as in (\ref{eq:LURm2}) the random component
\[
\sum\limits _{k=1}^{K_{n,m}}\dfrac{\sqrt{\lambda_{k}}}{\lambda_{k}c^{2}+1}\xi_{k}\sim N\left(0,\omega^{2}\sum\limits _{k=1}^{K_{n,m}}\dfrac{\lambda_{k}}{\left(\lambda_{k}c^{2}+1\right)^{2}}\right)
\]
has a finite variance 
\[
\omega^{2}\sum\limits _{k=1}^{K_{n,m}}\dfrac{\lambda_{k}}{\left(\lambda_{k}c^{2}+1\right)^{2}}\leq\omega^{2}\sum\limits _{k=1}^{\infty}\dfrac{\lambda_{k}}{\left(\lambda_{k}c^{2}+1\right)^{2}}\leq\omega^{2}\sum\limits _{k=1}^{\infty}\lambda_{k}=\frac{\omega^{2}}{2},
\]
and the orders in (\ref{eq:LURm1}) and (\ref{eq:LURm3}) are controlled
by an argument as in (\ref{eq:LUR_bound1}). We thus conclude
that the leading term in (\ref{eq:mLURKL}) is also $o_{a.s.}(1)$,
that is, $n^{-1/2}\widehat{c}_{t}^{(m)}=o_{a.s.}(1)$
uniformly for all $t\leq n$. It follows that $\sup\limits _{0\leq t\leq n}\left|n^{-1/2}\hat{f}_{t}^{(m)}-J_{c}\text{\ensuremath{\left(t/n\right)}}\right|=o_{a.s.}(1)$
in the expanded probability space and $n^{-1/2}\hat{f}_{\lfloor nr\rfloor}^{(m)}$
weakly converges to $J_{c}(r)$ in the original space. 
\end{proof}

\section{Additional Numerical Results\label{sec:Additional-Numerical-Results}}\label{sec:numerical}

Remarks \ref{rem:Surprise1} and \ref{rem:Surprise2} following Proposition \ref{pop:hpltu} predict
that when $\lambda=\mu n^{4}$ the residual from the HP filter will
retain a near explosive component involving the factor $\mathrm{e}^{cr}$ when $c>0$, which suggests that MSEs should 
be larger at the localizing coefficient $c=3$ than at $c=-3$, ceteris paribus.
This outcome is observed in Table \ref{tab:LUR} when
$n=100$ for quarterly data and when $n=300$ for monthly data, but is unclear in the larger sample sizes because the tuning parameter $\lambda$ is kept to $\lambda=1600$ and $\lambda=129600$ in Table \ref{tab:LUR}. 
In further confirmation of Remarks \ref{rem:Surprise1} and \ref{rem:Surprise2}, Table \ref{tab:Appendix1} reports MSE results for the same DGPs of quarterly data as
in Table \ref{tab:LUR}, but using the tuning parameter $\lambda=1.6\times10^{-5}n^{4}$. For $n=100$, we have $\lambda=1600$ and the results
in this case in Table \ref{tab:Appendix1} are the same as those in the corresponding cells of
Table \ref{tab:LUR}. The same holds for the HRF filter because $\lambda$ is irrelevant
in the autoregression. 

Consider the LUR case of DGP4. The MSE of the HP filter under
$c=3$ is the largest, followed by $c=-3$ which in turn exceed those of $c=0$. These outcomes are fully consistent with theory as the exponential factor $\mathrm{e}^{cr}$ is present in both the near explosive ($c>0$) and near stationary ($c<0$) cases in (\ref{eq:LUR_c}). The HP filter fails to completely catch the exponential factor effects because it removes only polynomial trends up to the third order. In contrast, when $c=0$ the exponential function factor is no longer present in (\ref{eq:LUR_c}), so that the HP filter MSE slightly improves when $c=0$ relative to $c=-3$. The HP MSEs have similar rankings over $c$ for DGPs 5, whereas the MSEs in DGPs 6 are primarily affected by the presence of a structural break. 
The MSEs of the bHP filter show much smaller differences between near explosive, unit root, and near stationary cases. These results provide confirmation of the robustness of bHP's capabilities in trend-cycle determination under LUR generating mechanisms.

\begin{table}[t]
\centering 
\caption{MSE of the Estimated Trends with $\lambda=1.6\times10^{-5}n^4$: LUR}\label{tab:Appendix1}
\small
\begin{tabular}{cc|rrrr|rrrr|rrrr}
\hline\hline
\multirow{2}{*}{DGP} & \multirow{2}{*}{$n$} & \multicolumn{4}{c}{$c=3$} & \multicolumn{4}{c}{$c=0$} & \multicolumn{4}{c}{$c=-3$} \\
    &         & HP            & 2HP     & bHP   & HRF  & HP   & 2HP   & bHP   & HRF  & HP  & 2HP    & bHP    & HRF  \\
                     \hline
\multirow{3}{*}{4} & 100 & 4.11  & 2.12  & 1.51 & 9.86  & 1.77  & 1.51  & 1.37 & 6.42  & 1.80  & 1.54  & 1.40 & 5.73  \\
                   & 200 & 4.25  & 2.93  & 1.79 & 9.53  & 3.29  & 2.71  & 1.75 & 7.76  & 3.29  & 2.73  & 1.76 & 7.23  \\
                   & 300 & 5.64  & 4.13  & 2.08 & 9.47  & 4.79  & 3.93  & 2.01 & 8.25  & 4.83  & 3.97  & 2.03 & 7.91  \\
                   \hline 
\multirow{3}{*}{5} & 100 & 4.58  & 2.23  & 1.55 & 16.09 & 2.26  & 1.62  & 1.37 & 17.87 & 2.28  & 1.65  & 1.41 & 17.73 \\
                   & 200 & 4.71  & 3.04  & 1.83 & 11.33 & 3.75  & 2.82  & 1.78 & 12.11 & 3.76  & 2.83  & 1.80 & 11.88 \\
                   & 300 & 6.11  & 4.24  & 2.12 & 10.29 & 5.25  & 4.03  & 2.07 & 10.80 & 5.28  & 4.07  & 2.08 & 10.72 \\
                   \hline 
\multirow{3}{*}{6} & 100 & 16.94 & 12.06 & 8.37 & 87.73 & 14.62 & 11.44 & 8.04 & 94.59 & 14.60 & 11.44 & 8.08 & 94.25 \\
                   & 200 & 16.12 & 12.11 & 6.38 & 41.51 & 15.12 & 11.88 & 6.27 & 43.47 & 15.10 & 11.84 & 6.28 & 43.17 \\
                   & 300 & 17.08 & 12.97 & 5.82 & 29.20 & 16.31 & 12.84 & 5.73 & 30.43 & 16.26 & 12.80 & 5.75 & 30.23 \\ 
                     \hline\hline
\end{tabular}
\end{table}

Additional simulation experiments are conducted to further assess boosting capabilities to improve performance. In particular, some of the experiments in Section \ref{sec:Simulations} are re-run using the same models but with heavy tailed innovations as the noise components in both the trend and cycle. Tables \ref{tab:I2 t5} and \ref{tab:LUR t5} report the simulation results under the same DGPs for quarterly data in Section \ref{sec:Simulations}, except that $e_t\sim i.i.d.~\sigma_e\sqrt{3/5}\cdot t(5)$ and $v_t\sim i.i.d.~\sqrt{3/5}\cdot t(5)$ to keep the variance the same as those from the Gaussian errors. We find that the MSEs under the $t(5)$ errors are very close to those reported in Tables \ref{tab:I2} and \ref{tab:LUR}, showing the methods' robustness to heavy-tailed distributions.

\begin{table}[t]
\centering 
\caption{MSE of the Estimated Trend with $t(5)$ Errors: $I(2)$}\label{tab:I2 t5}
\small
\begin{tabular}{ccrrrr}
\multicolumn{6}{c}{Quarterly data}                                               \\
\hline\hline 
DGP                & n   & \multicolumn{1}{r}{HP} & \multicolumn{1}{r}{2HP} & \multicolumn{1}{r}{bHP} & \multicolumn{1}{r}{HRF} \\
\hline 
\multirow{3}{*}{1} & 100 & 26.66                  & 15.92                   & 12.81                       & 432.54                    \\
                   & 200 & 26.36                  & 15.66                   & 13.25                       & 718.46                    \\
                   & 300 & 26.21                  & 15.52                   & 13.64                       & 887.19                    \\
                   \hline 
\multirow{3}{*}{2} & 100 & 27.16                  & 16.04                   & 12.84                       & 452.22                    \\
                   & 200 & 26.38                  & 15.67                   & 13.25                       & 722.17                    \\
                   & 300 & 26.21                  & 15.52                   & 13.64                       & 887.59                    \\
                   \hline 
\multirow{3}{*}{3} & 100 & 39.49                  & 25.86                   & 20.57                       & 518.43                    \\
                   & 200 & 31.92                  & 20.13                   & 17.00                       & 761.15                    \\
                   & 300 & 29.86                  & 18.44                   & 16.15                       & 918.96             \\
                   \hline\hline                    
\end{tabular}
\end{table}

\begin{table}[htbp]
\centering 
\caption{MSE of the Estimated Trends with $t(5)$ Errors: LUR}\label{tab:LUR t5}
\small

\begin{tabular}{cc|rrrr|rrrr|rrrr}
\multicolumn{14}{c}{Quarterly data} \\   
\hline\hline
\multirow{2}{*}{DGP} & \multirow{2}{*}{$n$} & \multicolumn{4}{c}{$c=3$} & \multicolumn{4}{c}{$c=0$} & \multicolumn{4}{c}{$c=-3$} \\  
  &     & HP & 2HP    & bHP   & HRF  & HP & 2HP     & bHP   & HRF  & HP & 2HP     & bHP    & HRF  \\
                     \hline
                
\multirow{3}{*}{4} & 100 & 4.16  & 2.14  & 1.52 & 9.84  & 1.77  & 1.51  & 1.37 & 6.42  & 1.79  & 1.53  & 1.39 & 5.71  \\
                   & 200 & 1.80  & 1.50  & 1.43 & 9.52  & 1.79  & 1.51  & 1.45 & 7.79  & 1.81  & 1.53  & 1.47 & 7.27  \\
                   & 300 & 1.79  & 1.50  & 1.49 & 9.48  & 1.80  & 1.51  & 1.51 & 8.24  & 1.81  & 1.53  & 1.52 & 7.91  \\
                     \hline
\multirow{3}{*}{5} & 100 & 4.64  & 2.25  & 1.55 & 16.07 & 2.25  & 1.62  & 1.37 & 17.81 & 2.27  & 1.64  & 1.39 & 17.47 \\
                   & 200 & 1.82  & 1.50  & 1.43 & 11.27 & 1.81  & 1.52  & 1.44 & 12.13 & 1.82  & 1.53  & 1.47 & 11.89 \\
                   & 300 & 1.79  & 1.51  & 1.49 & 10.32 & 1.80  & 1.51  & 1.50 & 10.82 & 1.81  & 1.53  & 1.52 & 10.71 \\
                     \hline
\multirow{3}{*}{6} & 100 & 16.96 & 12.04 & 8.38 & 87.54 & 14.57 & 11.40 & 8.03 & 94.33 & 14.64 & 11.46 & 8.06 & 94.05 \\
                   & 200 & 7.48  & 6.04  & 4.72 & 41.35 & 7.47  & 6.05  & 4.73 & 43.42 & 7.51  & 6.09  & 4.75 & 43.40 \\
                   & 300 & 5.46  & 4.44  & 3.73 & 29.23 & 5.45  & 4.44  & 3.72 & 30.29 & 5.49  & 4.47  & 3.74 & 30.28 \\ 
                    \hline\hline
\end{tabular}
\end{table}

\end{document}